\documentclass[11pt]{article}
\usepackage{amsfonts,amsmath,amssymb,amsthm}
\usepackage[hmargin=1in,vmargin=1in]{geometry}
\usepackage{natbib}
\usepackage{enumerate}
\usepackage{caption}
\usepackage{footmisc,fancyvrb} 
\usepackage[colorlinks,citecolor=blue!80!black, linkcolor=red!60!violet,urlcolor=blue!80!black]{hyperref}
\usepackage{makeidx}
\usepackage{eurosym}
\usepackage{amsthm}
\usepackage{epsfig,lscape,rotate}
\usepackage{subcaption}
\usepackage{graphicx} 
\usepackage{amsmath}
\usepackage{epstopdf}
\usepackage{lscape}
\usepackage{rotating}
\usepackage{amssymb}
\usepackage{amsfonts}
\usepackage{euscript}
\usepackage{mathrsfs}
\usepackage{natbib}
\usepackage{multirow}
\usepackage{geometry}
\usepackage{algorithm}
\usepackage{booktabs}
\usepackage[usenames]{color}
\usepackage{tabularx}
\newcolumntype{Y}{>{\centering\arraybackslash}X}
\usepackage[normalem]{ulem}
\VerbatimFootnotes
\def\qed{\rule{2mm}{2mm}}

\usepackage{tikz,graphicx,pgfplots,pgfkeys,subcaption}   
\usetikzlibrary{datavisualization.formats.functions}
\pgfplotsset{compat=1.12} 

\def\addlegendimage{\csname pgfplots@addlegendimage\endcsname}

\usepackage[pagewise,mathlines]{lineno}
\mathchardef\dash="2D
\newtheorem{theorem}{Theorem}[section]

\newtheorem{assumption}{Assumption}[section]
\theoremstyle{definition}

\newtheorem{remark}{Remark}[section]
\usepackage{etoolbox} 
\AtEndEnvironment{remark}{~\qed}%

\begin{document}

\author{
Federico A. Bugni\\
Department of Economics\\
Northwestern University\\
\url{federico.bugni@northwestern.edu}
\and
Jia Li\\
School of Economics\\
Singapore Management University\\
\url{jiali@smu.edu.sg}
\and
Qiyuan Li\\
School of Economics\\
Singapore Management University\\
\url{qyli.2019@phdecons.smu.edu.sg}
}

\title{Permutation-based tests for discontinuities in event studies\thanks{We thank the Coeditor and two anonymous referees for comments and suggestions that have greatly improved the manuscript.}}

\maketitle
\vspace{-1.5cm}
\begin{abstract}
We propose using a permutation test to detect discontinuities in an underlying economic model at a known cutoff point. Relative to the existing literature, we show that this test is well suited for event studies based on time-series data. The test statistic measures the distance between the empirical distribution functions of observed data in two local subsamples on the two sides of the cutoff. Critical values are computed via a standard permutation algorithm. Under a high-level condition that the observed data can be coupled by a collection of conditionally independent variables, we establish the asymptotic validity of the permutation test, allowing the sizes of the local subsamples to be either be fixed or grow to infinity. In the latter case, we also establish that the permutation test is consistent. We demonstrate that our high-level condition can be verified in a broad range of problems in the infill asymptotic time-series setting, which justifies using the permutation test to detect jumps in economic variables such as volatility, trading activity, and liquidity. These potential applications are illustrated in an empirical case study for selected FOMC announcements during the ongoing COVID-19 pandemic.
\end{abstract}

\thispagestyle{empty}



\newpage
\setcounter{page}{1}


\section{Introduction\label{sec-intro}}

Many econometric problems can be expressed in terms of the continuity or the discontinuity of certain component in the underlying economic model. In an influential
paper, \cite{chow1960} tested the temporal stability in the demand for
automobiles, and subsequently stimulated a large literature on structural
breaks in time series analysis; see, for example, \cite{andrews1993}, \cite%
{stock1994}, \cite{baiperron1998}, and many references therein. In
microeconometrics, the regression discontinuity design (RDD) has been
extensively used for causal inference. This literature identifies and estimates an average treatment effect by evaluating discontinuities of conditional expectation functions of outcome and
treatment variables at a cutoff point of the running variable; see \cite%
{imbenslemieux2008} and \cite{leelemieux2010} for comprehensive reviews.%
\footnote{%
Coincidentally, the RDD was first proposed by 
\cite{TC1960} around the same time as the Chow test.} Meanwhile, a more
recent high-frequency financial econometrics literature has been devoted to
studying discontinuities, or jumps, in various financial time series (e.g.,
price, volatility, trading activity, etc.). The high-frequency jump
literature is pioneered by \cite{BNS06}, who propose the first nonparametric
test for asset price jumps using high-frequency data in an infill asymptotic
setting. More recently, \cite{fomc} study the jumps of volatility and
trading intensity in high-frequency jump regressions (\cite{jur}) that
closely resemble the classical RDD.

Although these strands of literature involve apparently different
terminology and technical tools, they share a common theme: The econometric
goal is to learn about differences in the data generating processes between two subsamples separated by the
cutoff. \cite{imbens2011} emphasize that these subsamples should be
``local'' to the cutoff point, which is
quite natural given the nonparametric nature of discontinuity inference (%
\cite{hahn2001}). The issue under study is thus a local version of the
classical two-sample problem. Correspondingly, the related inference is
often carried out using nonparametric two-sample t-tests, which are based on
kernel regressions in the RDD (\cite{hahn2001}, \cite{imbens2011}, \cite{CCT2014}%
) or, in the same spirit, spot high-frequency estimators (\cite{FN1996}, 
\cite{comterenault1998}, \cite{jacodprotter2012}, \cite{jur}, \cite{fomc})
in the infill time-series setting.

In an ideal scenario in which the subsamples separated by the cutoff are i.i.d.\ and independent of each other, the permutation test is an excellent tool to detect differences in their distributions. In particular, standard results for randomization inference (\citet[Chapter 15.2]{lehmannromano2005}) indicate that a permutation test implemented with any arbitrary test statistic is finite-sample valid under these conditions. The recent literature has investigated the properties of permutation tests {to detect differences between two samples} under less ideal conditions. One example is \cite{canaykamat2017}, who consider an RDD and show that permutation-based inference is asymptotically valid to detect discontinuities in the distribution of the baseline covariates at the cutoff. These authors implement their test with a finite number of observations that are located closest to the cutoff, effectively forcing them to concentrate on a small neighborhood of the cutoff as the sample size grows. In the same spirit, \cite{cattaneo/Tit/VB:17} propose using permutation-based inference to detect discontinuities at the cutoff under the ``local randomization framework'' introduced in \cite{cattaneo/etal:15}. Outside of the RDD literature, \cite{chungromano2013} and \cite{diciccio/romano:2017} investigate the asymptotic properties of permutation-based inference to test for differences in specific distributional features of two samples, such as the mean or the correlation coefficient. It is important to note that all of the references mentioned in this paragraph presume cross-sectional data.


{In the context of time-series applications, there is an active literature on change-point tests implemented via permutations. This approach was first suggested by \cite{antoch/huskova:2001} and later pursued by other authors. See \cite{huskova:2004} or \cite{horvath/rice:2014} for surveys of this literature. While most of this literature imposes independent errors, some allow for limited forms of weak dependence ({\cite{kirch/steinebach:2006}, \cite{kirch:2007}, and \cite{jentsch/pauly:2015}}). In contrast, our econometric setting accommodates essentially unrestricted persistence and nonstationarity in the underlying state processes (e.g., volatility), which better suits our interest on their dynamics over short time windows around economic news events. In the context of machine-learning methods, \cite{chernozhukov/wuthrich/zhu:2018} propose using permutations to implement conformal inference that allows for time series data.}

Set against this background, our main goal in this paper is to establish a
general theory for permutation-based discontinuity tests, with a special
emphasis on event studies based on time-series data. To capture the
``local'' nature of this problem, we adopt
an infill asymptotic framework, under which the inference concentrates on
observations ``close'' to the event time. Specifically, we consider the Cram\'{e}r-von Mises test statistic formed as the squared $L_{2}$ distance
between the empirical cumulative distribution functions for the two local
subsamples near the cutoff, and compute the critical value via a standard
permutation algorithm. As explained earlier, if the data were i.i.d., the behavior of this permutation test would follow directly from standard results for randomization inference. This ``off-the-shelf'' theory, however, is not applicable here
because time-series data observed in a short event window can be
serially highly dependent.

The main theoretical contribution of the present paper is to establish the
asymptotic validity of the permutation test in this non-standard setting. The theory has
two components. The first is a new general
result for permutation test. Specifically, we link the (feasible)
permutation test formed using the original data with an infeasible test
constructed in a ``coupling'' problem that involves
conditionally i.i.d.\ coupling variables. Since the latter
resembles the classical two-sample problem, the infeasible test controls
size exactly under the coupling null hypothesis (i.e., coupling
variables in the two subsamples are homogeneous), and is consistent under the
complementary alternative hypothesis. Under a proper notion of coupling,
which is customized for the permutation test, we show that the
feasible test inherits the same asymptotic rejection properties from the
infeasible one.  Since this result is of independent theoretical
interest that is well beyond our subsequent analysis in the infill
time-series setting, we frame the theory under general high-level
conditions so as to facilitate other types of applications.

The second component of our analysis pertains to specializing the general result to the
infill time-series setting designed for event-study applications. The
event-study framework is particularly relevant for studying macroeconomic
and financial shocks, including monetary shocks triggered by FOMC announcements (\cite{cochranepiazzesi2002}, \cite{nakamura2018QJE}), or \textquotedblleft natural disasters\textquotedblright\ such as
the ongoing COVID-19 pandemic. Following \cite{lixiu2016} and \cite{fomc}, we model
observed data using a general state-space framework, in which the observations are
discretely sampled from a latent state process \textquotedblleft
contaminated\textquotedblright\ by random disturbances. This model has been used to model variables such as asset returns, trading volume, duration, and bid-ask
spread, and readily
accommodates both continuously and discretely valued variables. Under this
state-space model, the temporal discontinuity in the data's distribution is mainly driven by the jump of the latent state process (e.g.,
asset volatility, trading intensity, and propensity of informed trading), which can be
detected by the permutation test. Under easy-to-verify primitive conditions,
we construct coupling variables and apply the aforementioned general theory
to establish the permutation test's asymptotic validity. 

We recognize two advantages of the proposed permutation test in comparison
with the standard approach based on the nonparametric ``spot'' estimation of the underlying state process.
Firstly, the permutation test attains asymptotic size control even if the
number of observations in each subsample is \emph{fixed}.\footnote{For similar type of results in the context of RDD; see \cite{cattaneo/etal:15}, \cite{cattaneo/Tit/VB:17}, \cite{canaykamat2017},  and \cite{bugni/canay:2021}.}
This remarkable property is reminiscent of the finite-sample exactness of the permutation test in the
classical two-sample problem for i.i.d.\ data. In contrast, the nonparametric estimation approach works in a fundamentally different way, as it relies on the asymptotic (mixed) normality
of the estimator, which in turn requires the sizes of the local subsamples to grow to infinity. In
empirical applications, however, it is often desirable to use a short time window,
either to reduce the effect of confounding factors in the background, or
simply because of the lack of observations soon after the occurrence of the
economic event (say, in a real-time research situation). Not surprisingly,
the conventional inference based on asymptotic Gaussianity often results in large
size distortions in this \textquotedblleft small-sample\textquotedblright\
scenario, as we demonstrate concretely in a realistically calibrated Monte
Carlo experiment (see Section \ref{sec:mc}). Meanwhile, the permutation test exhibits much more robust size control in finite samples.

The second advantage of the permutation test is its versatility: The same test can be applied in many different empirical contexts without any modification. On the other hand, the nonparametric estimation approach often relies on specific features of the problem, and needs to be designed on a case-by-case basis. Therefore, the proposed permutation test may be particularly attractive in new empirical environments for which tests based on the conventional approach are not yet developed or not yet well-understood. In Section \ref{sec:t2}, we illustrate this point more concretely in the context of testing for volatility jumps. In that case, the standard approach relies crucially on the assumption that the price shocks are Brownian in its design of the spot volatility estimator and the associated t-statistic, and it cannot be adapted easily to accommodate a more general setting with L\'{e}vy-driven shocks.\footnote{As explained by \cite{bns2001}, these more general processes offer the possibility of capturing important deviations from Brownian shocks and for flexible modelling of dependence structures. However, to the best of our knowledge, the estimation and inference of the spot volatility (i.e., the scaling process) in the non-Brownian case remains to be an open question in the literature. There is some limited work on the inference of integrated volatility functionals for the non-Brownian case (see \cite{todorov2012}) which demonstrates various distinct complications in the non-Brownian setting.} The permutation test, on the other hand, is valid even in the latter, more general, setting.

That being said, we stress that the proposed permutation test is a
complement, rather than substitute, for the conventional nonparametric estimation method,
because it has two limitations. One is that the permutation test focuses
exclusively on hypothesis testing, without producing a point estimate for
the jump of the state process (e.g., volatility) of interest, whereas the estimate is
a by-product of the conventional approach. In addition, the proposed
permutation test is purely nonparametric and it does not exploit any parametric structure that one may be willing to impose. It is therefore conceivable that in
certain semiparametric settings, more efficient tests may be designed to exploit a priori model restrictions. Put differently, the
aforementioned versatility of the permutation test may come with an efficiency
cost.  A better understanding about the robustness-efficiency tradeoff might
be an interesting topic for future research.

{In an empirical illustration, we apply the permutation test to a recent
sample of high-frequency intraday returns of the SPY ETF for the S\&P 500 index. Specifically, we focus on four important FOMC announcements during the ongoing COVID-19 pandemic, and test whether each announcement induces discontinuities in volatility, trading activity, and two measures of market illiquidity. We document robust empirical evidence for discontinuities in volatility and trading activity. We also find evidence for announcement-induced discontinuity in transaction cost (measured by bid-ask spread), but not in market impact (gauged by Amihud's measure).} This application highlights one of the main advantages of the proposed test, namely, it is applicable for a broad variety of high-frequency observations modeled in distinct ways, which is unlike, for example, the conventional t-test designed specifically for testing volatility jumps in the Brownian setting.

The rest of the paper is organized as follows. We present the asymptotic
theory for the permutation test in Section \ref{sec:t}. Section \ref{sec:mc}
reports the test's finite-sample performance in Monte Carlo experiments, and
Section \ref{sec:emp} presents the empirical illustration. Section \ref%
{conclusion} concludes. The appendix contains all proofs.

\emph{Notation}. We use $\left\Vert x\right\Vert $ to denote the Euclidean
norm of a vector $x$. For any real number $a$, we use $\lceil a \rceil$ to denote the smallest integer that is larger than $a$. For any constant $p\geq 1$, $\left\Vert \,\cdot
\,\right\Vert _{p}$ denotes the $L_{p}$ norm for random variables. For two
real sequences $a_{n}$ and $b_{n}$, we write $a_{n}\asymp b_{n}$ if $%
a_{n}/C\leq b_{n}\leq Ca_{n}$ for some finite constant $C\geq 1$.


\section{Theory\label{sec:t}}

\subsection{A general result for the asymptotic validity of permutation
tests \label{sec:t1}}

We first prove a new result that is broadly useful for establishing the
asymptotic validity of permutation tests. Because of its independent
theoretical interest, we develop the theory under high-level conditions. In
Section \ref{sec:t2}, below, we shall specialize this general result in
event-study applications under a more specific infill time-series setting,
for which the existing theory on permutation tests is not applicable.

Consider an array $(Y_{n,i})_{i\in \mathcal{I}_{n}}$ of $\mathbb{R}$-valued
observed variables defined on a probability space $\left( \Omega ,\mathcal{F}%
,\mathbb{P}\right) $, which may be either ``raw'' data or preliminary
estimators. Our econometric goal is to decide whether two subsamples $%
(Y_{n,i})_{i\in \mathcal{I}_{1,n}}$ and $\left( Y_{n,i}\right) _{i\in 
\mathcal{I}_{2,n}}$ have \textquotedblleft significantly\textquotedblright\
different distributions, where $(\mathcal{I}_{1,n},\mathcal{I}_{2,n})$ is a
partition of $\mathcal{I}_{n}\subseteq \mathbb{Z}$.
For ease of exposition, we assume that $%
\mathcal{I}_{1,n}$ and $\mathcal{I}_{2,n}$ contain the same number of
observations, denoted by $k_{n}$.\footnote{\label{foot:kn} All of our
results can be easily extended to the case when $\mathcal{I}_{1,n} $ and $%
\mathcal{I}_{2,n}$ have different sizes, but with the same order of
magnitude.} We stress from the outset that $k_{n}$ may either be fixed or
grow to infinity in the subsequent analysis. As such, our analysis speaks to
not only the classical finite-sample analysis of permutation tests, but also
the large-sample analysis routinely used in econometrics.

To implement the test, we first estimate the empirical cumulative
distribution functions (CDF) for the two subsamples using%
\begin{equation*}
\widehat{F}_{j,n}\left( x\right) \equiv \frac{1}{k_{n}}\sum_{i\in \mathcal{I}%
_{j,n}}1{\left\{ Y_{n,i}\leq x\right\} },\quad j\in \{1,2\}.
\end{equation*}%
We then measure their difference via the Cram\'{e}r--von Mises statistic,
given by%
\begin{equation}
\widehat{T}_{n}\equiv \frac{1}{2k_{n}}\sum_{i\in \mathcal{I}_{n}}\left( 
\widehat{F}_{1,n}\left( Y_{n,i}\right) -\widehat{F}_{2,n}\left(
Y_{n,i}\right) \right) ^{2}.
\label{eq:CVM}
\end{equation}%
For a significance level $\alpha \in \left( 0,1\right) $, we compute the
critical value via a standard permutation algorithm as in \citet[page
633]{lehmannromano2005}, which we specify in Algorithm 1 below. We use $\pi $
to denote a permutation of the elements of $\mathcal{I}_{n}$, that is, a
bijective mapping from $\mathcal{I}_{n}$ to itself. Let $G_{n}$ denote the
collection of all possible permutations of $\mathcal{I}_{n}$, with $M_{n}$
being its cardinality.

\medskip

\noindent \textbf{Algorithm 1}. Step 1. For each permutation $\pi \in G_n $,
compute the permuted test statistic $\widehat{T}_{n}(\pi )$ as $\widehat{T}%
_{n}$, but with $(Y_{n,i})_{i\in \mathcal{I}_{n}}$ replaced by $(Y_{n,\pi
\left( i\right) })_{i\in \mathcal{I}_{n}}$.

\noindent Step 2. Order $\{\widehat{T}_{n}(\pi ):\pi \in G_{n}\}$ as $%
\widehat{T}_{n}^{(1)}\leq \cdots \leq \widehat{T}_{n}^{(M_{n})}$. Set $%
\widehat{T}_{n}^{\ast }=\widehat{T}_{n}^{(k)}$ for $k=\left\lceil
M_{n}(1-\alpha )\right\rceil $.

\noindent Step 3. If $\widehat{T}_{n}>\widehat{T}_{n}^{\ast }$, reject the
null hypothesis. If $\widehat{T}_{n}<\widehat{T}_{n}^{\ast }$, do not reject
the null hypothesis. If $\widehat{T}_{n}=\widehat{T}_{n}^{\ast }$, reject
the null hypothesis with probability $\hat{p}_{n}\equiv (M_{n}\alpha -%
\widehat{M}_{n}^{+})/\widehat{M}_{n}^{0}$, where $\widehat{M}_{n}^{+}$ and $%
\widehat{M}_{n}^{0}$ are the cardinalities of $\{j:\widehat{T}_{n}^{(j)}>%
\widehat{T}_{n}^{\ast }\}$ and $\{j:\widehat{T}_{n}^{(j)}=\widehat{T}%
_{n}^{\ast }\}$, respectively. The resulting test then rejects according to the critical function $%
\hat{\phi}_{n}\equiv 1\{\widehat{T}_{n}>\widehat{T}_{n}^{\ast }\}+\hat{p}%
_{n}1\{\widehat{T}_{n}=\widehat{T}_{n}^{\ast }\}$.\hfill $\Box $

\medskip

\begin{remark}
\label{rem:RandomPerm} The test specified in Algorithm 1 is
a randomized test and has a random outcome when $\widehat{T}_{n}=\widehat{T}%
_{n}^{\ast }$ { (i.e., it rejects the null hypothesis with probability $\hat{p}_n$ independent of the data); see Section 3.1 in \cite{lehmannromano2005} for details about randomized tests.}  One can construct a non-randomized (and more conservative)
version by replacing $\hat{p}_{n}$ with zero. Also, in practice, $M_{n}$ may
be too large to consider $G_n$ in its entirety. In such cases, we could
replace ${G}_n$ with a random subset of it, denoted by $\widehat{G}_n$, and
composed of the identity permutation and an i.i.d.\ sample of permutations
in ${G}_n$. All of the formal results in this paper apply if we use $%
\widehat{G}_n$ instead of ${G}_n$ in Algorithm 1. 
\end{remark}

\begin{remark}\label{rem:otherStat}
In this paper, we use the Cram\'er-von Mises statistic in \eqref{eq:CVM} for concreteness and simplicity of exposition. However, the supplement of this paper shows that our main results extend well beyond this particular statistic. In particular, the asymptotic validity of the permutation test extends to any other rank statistic, i.e., a statistic that only depends on the rank of the observations. Furthermore, our consistency result applies to any other rank statistic under mild regularity conditions. For example, both of these results hold if we replace \eqref{eq:CVM} with the Kolmogorov-Smirnov statistic, given by
\begin{equation*}
\widehat{T}_{n}\equiv \max_{i\in \mathcal{I}_{n}}\left| 
\widehat{F}_{1,n}\left( Y_{n,i}\right) -\widehat{F}_{2,n}\left(
Y_{n,i}\right) \right|.
\end{equation*}
For details, see the supplement of this paper.
\end{remark}

If the data $(Y_{n,i})_{i\in \mathcal{I}_{n}}$ are i.i.d., then the null
hypothesis of the classical two-sample problem holds, and \citet[Theorem
15.2.1]{lehmannromano2005} implies that the aforementioned permutation test
has \emph{exact} size control in finite samples. This is a remarkable
property of the permutation test, as it holds without requiring any specific
distributional assumptions on the data. In contrast to the classical
two-sample problem, however, we shall not assume that the data are
independent, or even \textquotedblleft weakly\textquotedblright\ dependent
(e.g., mixing). As mentioned in the Introduction, the main goal of this
paper is to study the permutation test for time-series data observed within
a short event window (say, a few days or hours), which can be serially
highly dependent in practice. Our key theoretical insight is that the
permutation test is still asymptotically valid if the data $(Y_{n,i})_{i\in 
\mathcal{I}_{n}}$ can be approximated, or \textquotedblleft
coupled,\textquotedblright\ by another collection of variables that are
conditionally independent, as formalized by the following assumption.

\begin{assumption}
\label{as:c}There exists a collection of variables $\left( U_{n,i}\right)
_{i\in \mathcal{I}_{n}}$ such that the following conditions hold for a
sequence $(\mathcal{G}_{n})_{n\geq 1}$ of $\sigma $-fields: (i) for each $%
n\geq 1$, the variables $(U_{n,i})_{i\in \mathcal{I}_{n}}$ are $\mathcal{G}%
_{n}$-conditionally independent, and $U_{n,i}$ has the same $\mathcal{G}_{n}$%
-conditional distribution as $U_{n,j}$ if $i,j$ belong to the same subsample
(i.e., $\mathcal{I}_{1,n}$ or $\mathcal{I}_{2,n}$); (ii) for any real
sequence $\eta _{n}=o(1)$, we have $\sup_{x\in \mathbb{R}}\mathbb{P}\left(
\left\vert U_{n,i}-x\right\vert \leq \eta _{n}|\mathcal{G}_{n}\right)
=O_{p}\left( \eta _{n}\right) $; (iii) $\max_{i\in \mathcal{I}_{n}}|%
\widetilde{Y}_{n,i}-U_{n,i}|=o_{p}( k_{n}^{-2}) $, where $(\widetilde{Y}%
_{n,i})_{i\in \mathcal{I}_{n}}$ is an identical copy of $(Y_{n,i})_{i\in 
\mathcal{I}_{n}}$ in $\mathcal{G}_{n}$-conditional distribution.
\end{assumption}

Assumption \ref{as:c} lays out the high-level structure for bridging our
analysis with the classical theory on permutation tests, which we carry out
in Theorem \ref{thm1} below. Condition (i) sets up the \textquotedblleft
coupling\textquotedblright\ problem, which corresponds to a conditional
version of the classical two-sample problem, treating the $(U_{n,i})_{i\in 
\mathcal{I}_{1,n}}$ and $(U_{n,i})_{i\in \mathcal{I}_{2,n}}$ variables as
\textquotedblleft data\textquotedblright. In part (a) of Theorem \ref{thm1}%
, we consider the situation in which both subsamples have the same
conditional distribution. In this case, our coupling variables $%
(U_{n,i})_{i\in \mathcal{I}_{n}}$ give rise to an infeasible permutation
test that can be analyzed as a classical two-sample problem. In particular,
this infeasible permutation test attains the exact finite-sample size under
our conditions.

This infeasible test, however, only plays an auxiliary role in our analysis,
because our interest is on the feasible test $\hat{\phi}_{n}$ formed using
the original $(Y_{n,i})_{i\in \mathcal{I}_{n}}$ data. Therefore, a key
component of our theoretical argument in Theorem \ref{thm1} is to show that
the feasible test for the original data inherits asymptotically the same
rejection properties from the infeasible test. Conditions (ii) and (iii) in
Assumption \ref{as:c} are introduced for this purpose. Specifically,
condition (ii) requires the variable $U_{n,i}$ to be non-degenerate, in the
sense that its conditional probability mass within any small $\left[ x-\eta
,x+\eta \right] $ interval is of order $O\left( \eta \right) $ in
probability.\footnote{%
{Condition (ii) is satisfied if $%
U_{n,i}$ has a conditional probability density that is uniformly bounded in probability.}}
Condition (iii) specifies the requisite approximation accuracy of the
coupling variables. We note that this condition is easier to hold when $k_n$ is smaller, because the joint coupling requirement would involve a smaller number of variables and the $o_p(k_n^{-2})$ error bound is easier to attain. This condition can be verified under more primitive
conditions pertaining to the smoothness of underlying processes and {an upper bound on the
growth rate of $k_{n}$}, as detailed in Section \ref{sec:t2}.\footnote{%
In our applications, we can often verify condition (iii) with $\widetilde{Y}%
_{n,i}=Y_{n,i}$. Nonetheless, allowing $\widetilde{Y}_{n,i}\neq Y_{n,i}$ is
useful when $Y_{n,i}$ is itself an estimator. For example, if $%
(Y_{n,i})_{i\in \mathcal{I}_{n}}$ is a finite collection of estimators that
converge jointly in distribution, then the coupling required in Assumption \ref{as:c}(iii) can be obtained via Skorokhod representation. As  such, the theory developed in  \cite{canaykamat2017} can be cast in our framework with $\mathcal{G}_n$ being the trivial information set and $k_n$ being a fixed constant, although our anti-concentration condition in Assumption \ref{as:c}(ii) is slightly stronger than the continuity condition in \citet[Assumption 4.2]{canaykamat2017}.}

\begin{theorem}
\label{thm1}Under Assumption \ref{as:c}, the following statements hold for
the permutation test $\hat{\phi}_{n}$ described in Algorithm 1:

(a) If the variables $(U_{n,i})_{i\in \mathcal{I}_{n}}$ have the same $\mathcal{G}_{n}$%
-conditional distribution, we have $\mathbb{E}[\hat{\phi}_{n}]\rightarrow \alpha 
$.

(b) Let $Q_{j,n}\left( \cdot \right) $ denote the $\mathcal{G}_{n}$%
-conditional distribution function of $U_{n,i}$ for $i\in \mathcal{I}_{j,n}$
and $j\in \left\{ 1,2\right\} $, and $\overline{Q}_{n}=(Q_{1,n}+Q_{2,n})/2$.
If $k_{n}\rightarrow \infty $ and $\mathbb{P}(\int \left(
Q_{1,n}(x)-Q_{2,n}(x)\right) ^{2}d\overline{Q}_{n}\left( x\right) >\delta
_{n})\rightarrow 1$ for any real sequence $\delta _{n}=o(1)$, we have $%
\mathbb{E}[\hat{\phi}_{n}]\rightarrow 1$.
\end{theorem}

Theorem \ref{thm1} characterizes the asymptotic rejection probabilities of
the feasible test $\hat{\phi}_{n}$ under the null and alternative hypotheses
of the two-sample problem for the coupling variables. Part (a) pertains to
the situation in which the two subsamples of coupling variables, $%
(U_{n,i})_{i\in \mathcal{I}_{1,n}}$ and $(U_{n,i})_{i\in \mathcal{I}_{2n}}$,
have the same conditional distribution, which corresponds to the null
hypothesis. In this case, the theorem shows that the asymptotic rejection
probability of the feasible test is equal to the nominal level $\alpha $. It
is relevant to note that this result holds whether $k_n$ is fixed or
divergent. This property is clearly reminiscent of the permutation test's
finite-sample exactness in the classical setting.

Part (b) of Theorem \ref{thm1} concerns the power of the feasible test $\hat{%
\phi}_{n}$. It shows that the feasible test rejects with probability
approaching one when the conditional distributions of the two coupling
subsamples, $Q_{1,n}$ and $Q_{2,n}$, are different, in the sense that their
\textquotedblleft distance\textquotedblright\ measured by $\int \left(
Q_{1,n}(x)-Q_{2,n}(x)\right) ^{2}d\overline{Q}_{n}\left( x\right) $ is
asymptotically non-degenerate, where the mixture distribution $\overline{Q}%
_{n}$ captures approximately the distribution of the permuted data.\footnote{When the two subsamples have different sizes, the same result goes through with $\overline{Q}_n$ defined as the sample-size weighted average between $Q_{1,n}$ and $Q_{2,n}$, {given by $\overline{Q}_n=(Q_{1,n}|\mathcal{I}_{1,n}| + Q_{2,n}|\mathcal{I}_{2,n}|)/(|\mathcal{I}_{1,n}|+|\mathcal{I}_{2,n}|)$.}} This
consistency-type result requires that the information available from each
subsample grows with the sample size, i.e., $k_{n}\rightarrow \infty $. This
result appears to be new in the context of permutation-based tests under a
fixed alternative for the coupling variables. In particular, we note that an
analogous result is unavailable in \cite{canaykamat2017}, as they restrict
attention to an asymptotic framework with a fixed $k_n$, which makes a
consistency-type result unavailable. Our proof relies on applying %
\citet[Theorem 15.2.3]{lehmannromano2005} to the infeasible test, for which
we use the coupling construction developed by \cite{chungromano2013} to show
that the so-called \cite{hoeffding:1952} condition is satisfied. We note
that this argument is used to establish the consistency of the permutation
test rather than its asymptotic size property.

Theorem \ref{thm1} establishes the relation between the rejection
probability of the feasible test $\hat{\phi}_{n}$ and the homogeneity (or
the lack of it) across the two coupling subsamples $(U_{n,i})_{i\in \mathcal{%
I}_{1,n}}$ and $(U_{n,i})_{i\in \mathcal{I}_{2,n}}$. This result does not
speak directly to hypotheses formulated in terms of the original $\left(
Y_{n,i}\right) _{i\in \mathcal{I}_{n}}$ observations. Rather, its
theoretical significance is to \textquotedblleft absorb\textquotedblright\
all generic technicalities stemming from the (feasible) permutation test,
which in turn considerably simplifies our overall analysis. The residual
issue for any specific application is to explicitly construct the coupling
variables and translate their homogeneity in terms of the primitive
structures of the original empirical problem, which can be done using
domain-specific techniques. We provide general results along this line in
the infill time-series context, as detailed in Section \ref{sec:t2} below.

To help anticipate the general discussion, it is instructive to sketch the
scheme in a basic running example. Let $Y_{n,i}=\Delta _{n}^{-1/2}(P_{\left(
i+1\right) \Delta _{n}}-P_{i\Delta _{n}})$ be the scaled increment of the
asset price process $P_{t}$ over the $i$th sampling interval.
Let $\tau $ be a \textquotedblleft cutoff\textquotedblright\ time point of
interest that is known to the researcher (e.g., the announcement time of a news release) and $i^{\ast }$ be
the unique integer such that $\tau \in \lbrack i^{\ast }\Delta
_{n},(i^{\ast }+1)\Delta _{n})$.\footnote{%
The integer $i^{\ast }$ depends on $n$. We
suppress this dependence in our notation for simplicity and to avoid having nested subscripts.}
We consider two index sets $\mathcal{I}_{1,n}=\{i^{\ast }-k_{n},\ldots
,i^{\ast }-1\}$ and $\mathcal{I}_{2,n}=\{i^{\ast }+1,\ldots ,i^{\ast
}+k_{n}\}$, which collect observations before and after the cutoff,
respectively. We consider an asymptotic setting in which these subsamples
are \textquotedblleft local\textquotedblright\ in calendar time, that is, $%
k_{n}\Delta _{n}\rightarrow 0$. Note that this implies that $\Delta
_{n}\rightarrow 0$, which means that we are considering an infill asymptotic
setting. If $P_{t}$ is an It\^{o} process with respect to an information
filtration $(\mathcal{F}_{t})_{t\geq 0}$, we may represent $Y_{n,i}$ as%
\begin{equation}
Y_{n,i}=\Delta _{n}^{-1/2}\int_{i\Delta _{n}}^{\left( i+1\right) \Delta
_{n}}b_{s}ds+\Delta _{n}^{-1/2}\int_{i\Delta _{n}}^{\left( i+1\right) \Delta
_{n}}\sigma _{s}dW_{s},\quad \text{for}\quad i\in \mathcal{I}_{n},
\label{ysimple}
\end{equation}%
where $b_{t}$ is the drift process, $\sigma _{t}$ is the stochastic
volatility process, and $W_{t}$ is a standard Brownian motion.\footnote{%
Note that $\mathcal{I}_{n}$ does not include the $i^{\ast }$th return
observation. Therefore, although the returns in (\ref{ysimple}) do not
contain price jumps, an event-induced price jump is allowed to occur at time 
$\tau $.} If the $\sigma _{t}$ process is smooth (e.g., H\"{o}lder
continuous) in a local neighborhood before $\tau $, then the volatility
throughout the pre-event subsample $\mathcal{I}_{1,n}$ is approximately $%
\sigma _{(i^{\ast }-k_{n})\Delta _{n}}$. Further recognizing that the drift
term is negligible relative to the Brownian component, we can approximate $%
Y_{n,i}$ for each $i\in \mathcal{I}_{1,n}$ using the coupling variables 
\begin{equation}
U_{n,i}=\sigma _{(i^{\ast }-k_{n})\Delta _{n}}\Delta _{n}^{-1/2}(W_{\left(
i+1\right) \Delta _{n}}-W_{i\Delta _{n}})\sim \mathcal{MN}\left( 0,\sigma
_{(i^{\ast }-k_{n})\Delta _{n}}^{2}\right) ,  \label{usimple}
\end{equation}%
where $\mathcal{MN}$ denotes the mixed normal distribution. Since the
Brownian motion has independent and stationary increments, it is easy to see
that the coupling variables $(U_{n,i})_{i\in \mathcal{I}_{1,n}}$ are $%
\mathcal{F}_{(i^{\ast }-k_{n})\Delta _{n}}$-conditionally i.i.d. Moreover,
if the volatility process $\sigma _{t}$ does not jump at the cutoff time $%
\tau $, we may follow the same logic to extend the approximation in (\ref%
{usimple}) further to $i\in \mathcal{I}_{2,n}$. In other words, if the
volatility process process does not jump then the coupling variables $%
(U_{n,i})_{i\in \mathcal{I}_{n}}$ are conditionally i.i.d., which
corresponds to the situation in part (a) of Theorem \ref{thm1}. On the other
hand, if the volatility process jumps at time $\tau $, say by a constant $%
c\neq 0$, then the coupling variables for the $\mathcal{I}_{2,n}$ subsample
will instead take the form $U_{n,i}=\left( \sigma _{(i^{\ast }-k_{n})\Delta
_{n}}+c\right) (W_{\left( i+1\right) \Delta _{n}}-W_{i\Delta _{n}})$. In
this case, the two subsamples of $U_{n,i}$'s have distinct conditional
distributions (i.e., mixed normal with different conditional variances),
corresponding to the scenario in part (b) of Theorem \ref{thm1}.

Within the context of this illustrative example, we can further clarify a
key feature of the proposed test that holds more generally. It is not aimed
at detecting \textquotedblleft small\textquotedblright\ time-variations in
the distribution of the observed data. In fact, by allowing the drift $b_{t}$
and the volatility $\sigma _{t}$ to be time-varying, a smooth form of
heterogeneity is \emph{always} built in. The test instead detects abrupt
changes, or discontinuities, in the evolution of the distribution, which can
be more plausibly associated with the \textquotedblleft
lumpy\textquotedblright\ information carried by the underlying economic
announcement, as emphasized by \cite{nakamura2018JEP}. Specifically in this
example, the asset returns are locally centered Gaussian (due to the
assumption that the price is an It\^{o} process), and hence, the temporal
discontinuity in the return distribution manifests itself as a volatility
jump. The empirical scope of our permutation test, however, is far beyond
volatility-jump testing depicted in this illustration, as we shall
demonstrate in the remainder of the paper.

\subsection{Permutation tests for discontinuities in event studies\label%
{sec:t2}}

We now specialize the general Theorem \ref{thm1} into an infill asymptotic
time-series setting that is particularly suitable for event studies. By
introducing a mild additional econometric structure, we shall establish the
asymptotic validity of the permutation test under more primitive conditions
that are easy to verify in a variety of concrete empirical settings. As in
the running example above, we consider an event occurring at time $\tau
\in [i^{\ast }\Delta _{n},(i^{\ast }+1)\Delta _{n})$, which separates two subsamples indexed by $\mathcal{I%
}_{1,n}=\{i^{\ast }-k_{n},\ldots ,i^{\ast }-1\}$ and $\mathcal{I}%
_{2,n}=\{i^{\ast }+1,\ldots ,i^{\ast }+k_{n}\}$, respectively. All limits in
the sequel are obtained under the infill asymptotic setting with $\Delta
_{n}\rightarrow 0$.

We suppose that the data are generated from an approximate state-space model
of the form%
\begin{equation}
Y_{n,i}=g\left( \zeta _{i\Delta _{n}},\epsilon _{n,i}\right) +R_{n,i},\quad
i\in \mathcal{I}_{n},  \label{yss}
\end{equation}%
where the state process $\zeta _{t}$ is c\`{a}dl\`{a}g, adapted to a
filtration $\mathcal{F}_{t}$, and takes values in an open set $\mathcal{Z}%
\subseteq \mathbb{R}^{\dim (\zeta )}$; $(\epsilon _{n,i})_{i\in \mathcal{I}%
_{n}}$ are i.i.d.\ random disturbances taking values in some (possibly
abstract) space $\mathcal{E}$; $g\left( \cdot ,\cdot \right) $ is a
\textquotedblleft smooth\textquotedblright\ transform; and $R_{n,i}$ is a
residual term that is negligible relative to the leading term $g\left( \zeta
_{i\Delta _{n}},\epsilon _{n,i}\right) $ in a proper sense detailed below. A
simpler version of this state-space model without the $R_{n,i}$ residual
term has been used by \cite{lixiu2016} and \cite{fomc}, among others, for
modeling market variables such as trading volume and bid-ask spread. By
introducing the $R_{n,i}$ {residual} term, we can use a unified framework to
accommodate a broader class of models, which in particular include
increments of an It\^{o} semimartingale. We now revisit the model in (\ref%
{ysimple}) as the first illustration.

\bigskip

\noindent \textsc{Example 1 (Brownian Asset Returns)}. We represent the It%
\^{o}-process model (\ref{ysimple}) for asset returns in the form of (\ref%
{yss}) by setting $\zeta _{t}=\sigma _{t}$, $\epsilon _{n,i}=\Delta
_{n}^{-1/2}(W_{\left( i+1\right) \Delta _{n}}-W_{i\Delta _{n}})$, and $%
g\left( z,\epsilon \right) =z\epsilon $. The resulting residual term has the form%
\begin{equation}
R_{n,i}=\Delta _{n}^{-1/2}\int_{i\Delta _{n}}^{\left( i+1\right) \Delta
_{n}}b_{s}ds+\Delta _{n}^{-1/2}\int_{i\Delta _{n}}^{\left( i+1\right) \Delta
_{n}}(\sigma _{s}-\sigma _{i\Delta _{n}})dW_{s},\quad i\in \mathcal{%
I}_{n}.  \label{ex1:R}
\end{equation}%
Under mild and fairly standard regularity conditions, it is easy to show
that $\max_{i \in \mathcal{I}_n}|R_{n,i}|$ is $o_p(1)$. On the other hand, the
leading term $g\left( \zeta _{i\Delta _{n}},\epsilon _{n,i}\right) $ has a
non-degenerate centered mixed Gaussian distribution with conditional
variance $\sigma _{i\Delta _{n}}^{2}$.\hfill $\square $

\bigskip

This running example further illustrates the distinct roles played by $%
\zeta_t$, $\epsilon_{n,i}$, and $R_{n,i}$ in our state-space model (\ref{yss}%
). The leading term $g\left( \zeta _{i\Delta _{n}},\epsilon _{n,i}\right) $
captures the ``main feature'' of the observed data; in addition, since the $%
\epsilon _{n,i}$ disturbance terms are i.i.d., any \textquotedblleft
large\textquotedblright\ change in the empirical distribution across the two
subsamples must be attributed to the time-$\tau $ discontinuity in the state
process $\zeta _t$. From this description, it follows that the hypothesis
test for the continuity of the distribution of the main feature of the
observed data can be formulated as%
\begin{equation}
H_{0}:\Delta \zeta _{\tau }=0\quad \text{versus}\quad H_{a}:\Delta \zeta
_{\tau }\neq 0,  \label{hypo}
\end{equation}%
where $\Delta \zeta _{\tau }\equiv \zeta _{\tau }-\zeta _{\tau -}\equiv
\zeta _{\tau }-\lim_{s\uparrow \tau }\zeta _{s}$ denotes the jump of the
state process at time $\tau $.

With the state-space model (\ref{yss}) in place, we can design more
primitive sufficient conditions for establishing the asymptotic validity of
the permutation test under the hypotheses in (\ref{hypo}). We need some
additional notation to describe these conditions. For each fixed $z\in 
\mathcal{Z}$, let $f_{z}\left( \cdot \right) $ and $F_{z}\left( \cdot
\right) $ denote the probability density function (PDF) and the CDF of the
random variable $g\left( z,\varepsilon _{n,i}\right) $, respectively. It is
also convenient to introduce a \textquotedblleft shifted\textquotedblright\
version of $\zeta _{t}$ defined as $\tilde{\zeta}_{t}\equiv \zeta
_{t}-\Delta \zeta _{\tau }1\{t\geq \tau \}$, which has the same increments
as $\zeta _{t}$ over time intervals not containing $\tau $.

\begin{assumption}
\label{as:ind}(i) The collection of variables $(\epsilon _{n,i})_{i\in 
\mathcal{I}_{n}}$ are i.i.d.\ and, for each $k\in \mathcal{I}_{n}$, the
variables $(\epsilon _{n,i})_{i\geq k}$ are independent of $\mathcal{F}%
_{k\Delta _{n}}$. Moreover, for any compact subset $\mathcal{K}\subseteq 
\mathcal{Z}$, we have (ii) $\sup_{x\in \mathbb{R},z\in \mathcal{K}%
}f_{z}(x)<\infty $; and (iii) $\inf_{z\in \mathcal{K}}\int_{\mathbb{R}%
}\left( F_{z}(x)-F_{z+c}(x)\right) ^{2}dF_{z}(x)>0$ whenever $c\neq 0$.
\end{assumption}

\begin{assumption}
\label{as:smooth}There exist a sequence $\left( T_{m}\right) _{m\geq 1}$ of
stopping times increasing to infinity, a sequence of compact subsets $(%
\mathcal{K}_{m})_{m\geq 1}$ of $\mathcal{Z}$, and a sequence $\left(
K_{m}\right) _{m\geq 1}$ of constants such that for some real sequence $%
a_{n}\geq 1$ and each $m\geq 1$: (i) $\left\Vert g\left( z,\epsilon
_{n,i}\right) -g(z^{\prime },\epsilon _{n,i})\right\Vert _{2}\leq
K_{m}a_{n}\left\Vert z-z^{\prime }\right\Vert $ for all $z,z^{\prime }\in 
\mathcal{K}_{m}$; (ii) $\zeta _{t}$ takes values in $\mathcal{K}_{m}$ for
all $t\leq T_{m}$, and $\Vert \tilde{\zeta}_{t\wedge T_{m}}-\tilde{\zeta}%
_{s\wedge T_{m}}\Vert _{2}\leq K_{m}\left\vert t-s\right\vert ^{1/2}$ for
all $t,s$ in some fixed neighborhood of $\tau $; (iii) $\max_{i\in \mathcal{I%
}_{n}}\left\vert R_{n,i}\right\vert =o_{p}(k_{n}^{-2})$.
\end{assumption}

Assumption \ref{as:ind} entails regularity conditions pertaining to the
random disturbance terms, which are often easy to verify in concrete
examples as demonstrated later in this subsection. Assumption \ref{as:smooth}
imposes a set of smoothness conditions that permits the approximation of the
observed data using properly constructed coupling variables.\footnote{%
Note that the assumption is framed in a localized fashion using the stopping
times $(T_{m})_{m\geq 1}$, which is a standard technique for weakening the
regularity condition in the infill asymptotic setting. See \citet[Section
4.4.1]{jacodprotter2012} for a comprehensive discussion on the localization
technique.} Specifically, condition (i) requires that the random function $%
z\mapsto g\left( z,\epsilon _{n,i}\right) $ is Lipschitz in $z$ over compact
sets under the $L_{2}$ distance. The $a_{n}$ sequence captures the scale of
the Lipschitz coefficient. In many applications, we can verify this
condition simply with $a_{n}\equiv 1$, but allowing $a_{n}$ to diverge to
infinity is sometimes necessary (see Example 2 below). Condition (ii) states
that the $\zeta _{t}$ process is locally compact (up to each stopping time $%
T_{m}$) and, upon removing the fixed-time discontinuity at $\tau $, it is $%
(1/2)$-H\"{o}lder continuous under the $L_{2}$ norm. This H\"{o}%
lder-continuity requirement can be easily verified using well-known results
provided that the $\tilde{\zeta}$ process is an It\^{o} semimartingle or a
long-memory process (see \citet[Chapter 2]{jacodprotter2012} and \cite%
{liliu2020}). Condition (iii) imposes the requisite assumptions on the
residual terms. In some applications, this condition holds trivially with $%
R_{n,i}\equiv 0$, but, more generally, it needs to be verified on a
case-by-case basis using (relatively standard) infill asymptotic techniques.
Theorem \ref{thmc}, below, establishes the size and power properties of the
permutation test under the hypotheses described in (\ref{hypo}).

\begin{theorem}
\label{thmc} In the state-space model (\ref{yss}), suppose that Assumptions %
\ref{as:ind} and \ref{as:smooth} hold, and that $a_{n}k_{n}^{3}\Delta
_{n}^{1/2}=o(1)$. Then, the following statements hold for the permutation
test $\hat{\phi}_{n}$ described in Algorithm 1:

(a) Under the null hypothesis in (\ref{hypo}), i.e., $\Delta \zeta _{\tau }=0
$, we have $\mathbb{E}[\hat{\phi}_{n}]\rightarrow \alpha $;

(b) Under a fixed alternative hypothesis in (\ref{hypo}), i.e., $\Delta
\zeta _{\tau }=c$ for some (unknown) constant $c\neq 0$ and if $k_{n}\rightarrow \infty $, we have $\mathbb{E}[%
\hat{\phi}_{n}]\rightarrow 1$.
\end{theorem}

This theorem is proved by verifying the high-level conditions in Theorem \ref%
{thm1} with properly constructed coupling variables analogous to those in
equation (\ref{usimple}). The condition $a_{n}k_{n}^{3}\Delta _{n}^{1/2}=o(1)
$ mainly requires that the window size $k_{n}$ does not grow too fast, which
ensures the closeness between the coupling variables and the original data.
In the typical case with $a_{n}=1$, it reduces to $k_{n}=o(\Delta_{n}^{-1/6})$.\footnote{This sufficient condition for the growth rate of $k_n$ is different from the conditions needed for conventional asymptotic-Gaussian-based spot inference, which requires {$k_n \asymp \Delta_n^{-\iota}$ for some $\iota\in(0,1/2)$}. For the permutation test, $k_n$ may be fixed or grow to infinity. However, in the latter case, our condition on the growth rate of $k_n$ is more stringent than what is needed for the conventional spot inference theory.} In general, a larger $k_n$ allows one to utilize more data, but the associated longer event window may also lead to a larger nonparametric bias, and hence, a more severe size distortion. Part (a) shows that the permutation test attains the desired
asymptotic level under the null hypothesis in (\ref{hypo}). Again, we stress
that the test has valid asymptotic size control even in the
\textquotedblleft small-sample\textquotedblright\ case with fixed $k_{n}$.
As in Theorem \ref{thm1}, the \textquotedblleft
large-sample\textquotedblright\ condition $k_{n}\rightarrow \infty $ is
needed for establishing the consistency of the test under the alternative,
as shown in part (b).

In the remainder of this subsection, we use a few prototype examples to
demonstrate how the proposed test may be used in various empirical settings.
In particular, we show how to cast the specific problems into the
approximate state-space model (\ref{yss}), and discuss how to verify our
sufficient regularity conditions. We start by revisiting the running example.

\bigskip

\noindent \textsc{Example 1 (Brownian Asset Returns, Continued)}. Recall
that $\epsilon _{n,i}\equiv \Delta _{n}^{-1/2}(W_{\left( i+1\right) \Delta
_{n}}-W_{i\Delta _{n}})$, $\zeta _{t}=\sigma _{t}$, and $g\left( z,\epsilon
\right) =z\epsilon $. In this context, the hypothesis testing problem in (%
\ref{hypo}) represents a test of the continuity of the volatility process $%
\sigma _{t}$ at time $t=\tau$, i.e., 
\begin{equation*}
H_{0}:\Delta \sigma_{\tau }=0\quad \text{versus}\quad H_{a}:\Delta \sigma
_{\tau }\neq 0.
\end{equation*}
We suppose that the volatility process $\sigma _{t}$ is non-degenerate by
setting its domain to $\mathcal{Z}=\left( 0,\infty \right) $. Since the
Brownian motion has independent increments with respect to the underlying
filtration, the disturbance term $\epsilon _{n,i}$ satisfies Assumption \ref%
{as:ind}(i). In addition, for each point $z\in \mathcal{Z}$, the random
variable $f\left( z,\epsilon _{n,i}\right) $ has an $\mathcal{N}\left(
0,z^{2}\right) $ distribution. It is then easy to see that conditions (ii)
and (iii) in Assumption \ref{as:ind} hold for any compact subset $\mathcal{K}%
\subseteq \mathcal{Z}$ (note that $\mathcal{K}$ is necessarily bounded away
from zero). To verify Assumption \ref{as:smooth}, first note that $g\left(
z,\epsilon _{n,i}\right) -g(z^{\prime },\epsilon _{n,i})=(z-z^{\prime
})\epsilon _{n,i}$, and hence, $\left\Vert g\left( z,\epsilon _{n,i}\right)
-g(z^{\prime },\epsilon _{n,i})\right\Vert _{2}=|z-z^{\prime }|$. Assumption %
\ref{as:smooth}(i) thus holds for $a_{n}=1$. It is well known that $\sigma
_{t}$ is locally $(1/2)$-H\"{o}lder continuous under the $L_{2}$ norm if it
is an It\^{o} semimartingale or a long-memory process; if so, Assumption \ref%
{as:smooth}(ii) is satisfied if the $\sigma _{t}$ and $\sigma _{t}^{-1}$
processes are both locally bounded. Finally, to verify Assumption \ref%
{as:smooth}(iii), we assume that the drift process $b_{t}$ is locally
bounded. It is then easy to show via routine calculations that $\max_{i\in 
\mathcal{I}_{n}}|R_{n,i}|=O_{p}(k_{n}^{1/2}\Delta _{n}^{1/2})$. Since the
condition $a_{n}k_{n}^{3}\Delta _{n}^{1/2}=o(1)$ in Theorem \ref{thmc}
implies that $O_{p}(k_{n}^{1/2}\Delta _{n}^{1/2})=o_{p}(k_{n}^{-2})$, we
have $\max_{i\in \mathcal{I}_{n}}|R_{n,i}|=o_{p}(k_{n}^{-2})$ as needed in
Assumption \ref{as:smooth}(iii). All conditions in Theorem \ref{thmc} are
now verified, and this shows that the permutation test $\hat{\phi}_{n}$ is
asymptotically valid for testing the null hypothesis $\Delta \sigma _{\tau
}=0$.\hfill $\square $

\bigskip

Example 1 shows that the permutation test $\hat{\phi}_{n}$ is asymptotically
valid for testing the presence of a volatility jump. This is a relatively
familiar problem in the literature. It is therefore useful to contrast the
proposed permutation test with the standard approach, which is based on
nonparametric \textquotedblleft spot\textquotedblright\ estimators of the
asset price's instantaneous variances before and after the event time given
by, respectively, 
\begin{equation}
\hat{\sigma}_{\tau -}^{2}=\frac{1}{k_{n}}\sum_{i\in \mathcal{I}%
_{1,n}}Y_{n,i}^{2},\quad \hat{\sigma}_{\tau }^{2}=\frac{1}{k_{n}}\sum_{i\in 
\mathcal{I}_{2,n}}Y_{n,i}^{2}.  \label{spotvar}
\end{equation}%
Assuming $k_{n}\rightarrow \infty $ and $k_{n}^{2}\Delta _{n}\rightarrow 0$,
it can be shown that (see \citet[Chapter 13]{jacodprotter2012})%
\begin{equation}
\frac{k_{n}^{1/2}\left( \hat{\sigma}_{\tau }^{2}-\hat{\sigma}_{\tau
-}^{2}-(\sigma _{\tau }^{2}-\sigma _{\tau -}^{2})\right) }{\sqrt{2\hat{\sigma%
}_{\tau }^{4}+2\hat{\sigma}_{\tau -}^{4}}}\overset{d}{\longrightarrow }%
\mathcal{N}\left( 0,1\right) .  \label{ex1:tstat}
\end{equation}%
Thus, we can test $H_0: \Delta \sigma _{\tau }=0$ by comparing the
t-statistic $k_{n}^{1/2}\left( \hat{\sigma}_{\tau }^{2}-\hat{\sigma}_{\tau
-}^{2}\right) /\sqrt{2\hat{\sigma}_{\tau }^{4}+2\hat{\sigma}_{\tau -}^{4}}$
with critical values based on the standard normal distribution.

Two remarks are in order. First, note that the asymptotic size control of
the standard approach relies on the asymptotic normal approximation (\ref%
{ex1:tstat}), which depends crucially on $k_{n}\rightarrow \infty $ (in
addition to having $\Delta_n\to 0$) because the underlying central limit
theorem is obtained by aggregating a \textquotedblleft
large\textquotedblright\ number of martingale differences. Hence, the t-test
may suffer from severe size distortion when $k_{n}$ is relatively small.
This issue is empirically relevant because an applied researcher may use a
short time window to capture short-lived ``impulse-like'' dynamics and/or to
minimize the impact of other confounding economic factors in the background.
Moreover, for \textquotedblleft real-time\textquotedblright\ applications,
the researcher may have no choice but to use a small $k_{n}$ simply because
of the limited amount of available data soon after the event time $\tau $.
In sharp contrast, the permutation test controls asymptotic size even when $%
k_{n}$ is fixed. This remarkable property is inherited from the coupling
two-sample problem, in which the permutation test controls size exactly
regardless of whether $k_{n}$ is fixed or grows to infinity.

The second and perhaps practically more important difference between the two
tests is that the permutation test is more versatile. Under the
spot-estimation-based approach, both the design of the spot estimators in (%
\ref{spotvar}) and the convergence in (\ref{ex1:tstat}) depend heavily on
the fact that the increments of the Brownian motion are not only i.i.d., but
also Gaussian. Gaussianity is obviously essential for the conventional
approach because, among other things, it ensures that the instantaneous
variance of the normalized returns are well-defined.\footnote{%
Recall that many distributions used in continuous-time models do not have
finite second moments. For example, within the class of stable
distributions, the Gaussian distribution is the only one with a finite
second moment. Moreover, Gaussianity also implies that the variance of $%
\Delta _{n}^{-1}(W_{i\Delta _{n}}-W_{(i-1)\Delta _{n}})^{2}$ is 2, which
explains the \textquotedblleft 2\textquotedblright\ factor in the
denominator of the t-statistic.} The permutation test, on the other hand,
only exploits the i.i.d.\ property of the Brownian shocks, without relying
on {their} Gaussianity. Therefore, the permutation test readily accommodates a
more general model for asset returns with L\'{e}vy shocks, as we demonstrate
in the following example.

\bigskip

\noindent \textsc{Example 2 (L\'{e}vy-driven Asset Returns)}. We generalize
the model in Example 1 by replacing the Brownian motion $W$ with a L\'{e}vy
martingale $L$, so that the asset return has the form%
\begin{equation*}
P_{\left( i+1\right) \Delta _{n}}-P_{i\Delta _{n}}=\int_{i\Delta
_{n}}^{\left( i+1\right) \Delta _{n}}b_{s}ds+\int_{i\Delta _{n}}^{\left(
i+1\right) \Delta _{n}}\sigma _{s}dL_{s}, \quad i\in \mathcal{%
I}_{n}.
\end{equation*}%
In this case, we define the random disturbance as $\epsilon _{n,i}\equiv
\Delta _{n}^{-1/\beta }(L_{\left( i+1\right) \Delta _{n}}-L_{i\Delta _{n}})$
for some constant $\beta \in (1,2]$. The more general normalizing sequence $%
\Delta _{n}^{-1/\beta }$ is used to ensure that $\epsilon _{n,i}$ has a
non-degenerate distribution. For instance, if $L$ is a stable process, we
take $\beta $ to be its jump-activity index, so that $\epsilon _{n,i}$ has a
centered stable distribution (recall that the Brownian motion is a stable
process with index $\beta =2$). We treat the value of $\beta $ as unknown.
Since the permutation test is scale-invariant with respect to the data, we
can nonetheless regard the normalized return $Y_{n,i}=\Delta _{n}^{-1/\beta
}(P_{\left( i+1\right) \Delta _{n}}-P_{i\Delta _{n}})$ as directly
observable (because tests implemented for $P_{\left( i+1\right) \Delta
_{n}}-P_{i\Delta _{n}}$ and $Y_{n,i}$ are identical). To apply our theory,
we represent $Y_{n,i}$ using the state-space model (\ref{yss}) with $\zeta
_{t}=\sigma _{t}$, $g\left( z,\epsilon \right) =z\epsilon $, and the
residual term given by 
\begin{equation*}
R_{n,i}=\Delta _{n}^{-1/\beta }\int_{i\Delta _{n}}^{\left( i+1\right) \Delta
_{n}}b_{s}ds+\Delta _{n}^{-1/\beta }\int_{i\Delta _{n}}^{\left( i+1\right)
\Delta _{n}}\left( \sigma _{s}-\sigma _{i\Delta _{n}}\right) dL_{s},\quad i\in \mathcal{%
I}_{n}.
\end{equation*}%
Recognizing that the scaled L\'{e}vy increments $(\epsilon _{n,i})_{i\in 
\mathcal{I}_{n}}$ are i.i.d., we can verify Assumptions \ref{as:ind} and \ref%
{as:smooth} using similar arguments as in Example 1 but with $a_{n}=\Delta
_{n}^{1/2-1/\beta }$, which depicts the rate at which $\Vert \epsilon
_{n,i}\Vert _{2}$ diverges. In particular, the condition $%
a_{n}k_{n}^{3}\Delta _{n}^{1/2}=o(1)$ requires $k_{n}$ to obey $%
k_{n}=o(\Delta _{n}^{(1/\beta -1)/3})$. Then, we can apply Theorem \ref{thmc}
to show that the permutation test $\hat{\phi}_{n}$ is asymptotically valid
for testing the discontinuity in the volatility process $\sigma _{t}$ at
time $\tau $, regardless of whether the driving L\'{e}vy process is a
Brownian motion or not. {To our knowledge, our test is the first in the literature that accommodates L\'evy-type shocks in the context of testing for volatility jumps.}\hfill $\square $

\bigskip

So far, we have illustrated the use of the permutation test for
high-frequency asset returns data. Under the settings of Examples 1 and 2,
the distributional change of asset returns is mainly driven by the time-$%
\tau $ discontinuity in volatility, and hence, the permutation test is
effectively a test for volatility jumps. Example 2, in particular,
highlights the versatility and robustness of the permutation test compared
with the conventional approach based on spot estimation. Going one step
further, we now illustrate how to apply the permutation test to other types
of economic variables.

\bigskip

\noindent \textsc{Example 3 (Location-Scale Model for Volume)}. Consider a
simple model for trading volume, under which the volume within the $i$th
sampling interval\ is given by $Y_{n,i}=\mu _{i\Delta _{n}}+v_{i\Delta
_{n}}\epsilon _{n,i}$. The $\mu_t$ location process captures the local mean,
or trading intensity, and the $v_t$ scale process captures the time-varying
heterogeneity in the order size. This location-scale model fits directly
into the state-space model (\ref{yss}) with $\zeta _{t}=(\mu _{t},v_{t})$, $%
g\left( (\mu ,v),\epsilon \right) =\mu +v\epsilon $, and $R_{n,i}\equiv 0$.
Let $\mathcal{F}_{t}$ be the filtration generated by the $\zeta _{t}$
process. If $\epsilon _{n,i}$ is independent of the $\zeta_t$ process and
has finite second moment and bounded PDF, then it is easy to verify
Assumptions \ref{as:ind} and \ref{as:smooth} with $a_{n}=1$. Theorem \ref%
{thmc} thus implies that the permutation test is valid for testing the
discontinuity in $\zeta _{t}=\left( \mu _{t},v_{t}\right) $ at time $\tau $%
.\hfill $\square $

\bigskip

The location-scale structure in Example 3 is by no means essential in
applications, because the permutation test is valid provided that the more
general conditions in Assumptions \ref{as:ind} and \ref{as:smooth} hold.
This illustration is pedagogically convenient, in that it permits a
straightforward verification of our high-level conditions. That being said,
this example does reveal a limitation of our theory developed so far. That
is, the data variable needs to be continuously distributed, as required in
Assumption \ref{as:ind}(ii) (which in turn is related to Assumption \ref%
{as:c}(ii)). Observed data in actual applications are invariably discrete,
but this continuous-distribution assumption is often deemed as a reasonable
approximation to reality. In some situations, however, the discreteness in
the data is more salient. For example, the trading volume of a relatively
illiquid asset may take values as small integer multiples of the lot size
(e.g., 100 shares).\footnote{%
This issue has become less important in the equity market as retail
investors can now trade a single share, or even a fractional share, of a
stock. However, the lot size is still relevant for less liquid assets such
as option contracts or for equity data from earlier sample periods.} This
motivates us to directly confront the discreteness in the data, as detailed
in the next subsection.

\subsection{Extension: the case with discretely valued data\label{sec:t3}}

The extension will be carried out in similar steps as the theory developed
above. We start with modifying the general result in Theorem \ref{thm1} to
accommodate discretely valued observations; we then specialize the general theory to a high-frequency setting under more primitive conditions. Recall that $Q_{j,n}\left( \cdot
\right) $ denotes the $\mathcal{G}_{n}$-conditional distribution function of
the coupling variable $U_{n,i}$ for $i\in \mathcal{I}_{j,n}$ and $j\in
\left\{ 1,2\right\} $, and $\overline{Q}_{n}=(Q_{1,n}+Q_{2,n})/2$.

\begin{theorem}
\label{thm2}Suppose that there exists a collection of variables $\left(
U_{n,i}\right) _{i\in \mathcal{I}_{n}}$ that satisfies Assumption \ref{as:c}%
(i) for some sequence $(\mathcal{G}_{n})_{n\geq 1}$ of $\sigma $-fields, and 
$\mathbb{P}(\widetilde{Y}_{n,i}\neq U_{n,i})=o\left( k_{n}^{-1}\right) $
uniformly in $i\in \mathcal{I}_{n}$ where $(\widetilde{Y}_{n,i})_{i\in 
\mathcal{I}_{n}}$ is an identical copy of $(Y_{n,i})_{i\in \mathcal{I}_{n}}$
in $\mathcal{G}_{n}$-conditional distribution. Then, the following
statements hold for the test $\hat{\phi}_{n}$ described in Algorithm 1:

(a) If the variables $(U_{n,i})_{i\in \mathcal{I}_{n}}$ have the same $\mathcal{G}_{n}$%
-conditional distribution, we have $\mathbb{E}[\hat{\phi}_{n}]\rightarrow \alpha 
$.

(b) If $k_{n}\rightarrow \infty $ and $\mathbb{P(}\int \left(
Q_{1,n}(x)-Q_{2,n}(x)\right) ^{2}d\overline{Q}_{n}\left( x\right) >\delta
_{n})\rightarrow 1$ for any real sequence $\delta _{n}=o(1)$, we have $%
\mathbb{E}[\hat{\phi}_{n}]\rightarrow 1$.
\end{theorem}

Theorem \ref{thm2} establishes exactly the same asymptotic properties for
the permutation test as Theorem \ref{thm1}, but under different conditions:
it does not impose the anti-concentration requirement for the coupling
variable (i.e., Assumption \ref{as:c}(ii)), and the \textquotedblleft
distance\textquotedblright\ between the observed data and the coupling
variable is measured by the probability mass of $\{\widetilde{Y}_{n,i}\neq
U_{n,i}\}$. These modifications seem natural for the discrete-data setting.

Next, we specialize the general result in Theorem \ref{thm2} to the
state-space model (\ref{yss}), starting with some motivating examples. The
first is an alternative model for the trading volume that explicitly
features discretely valued data, which shows an interesting contrast to
Example 3.

\bigskip

\noindent \textsc{Example 4 (Poisson Model for Volume)}. Let $Y_{n,i}$ be
the trading volume of an asset within the $i$th sampling interval. Following 
\cite{andersen1996}, we model the discretely valued volume using a Poisson
distribution with time-varying mean. To form a state-space representation,
let $(\epsilon _{n,i}(t))_{t\geq 0}$ be a copy of the standard Poisson
process on $\mathbb{R}_{+}$, independent across $i$, and let $\zeta _{t}$ be
the time-varying mean process independent of the $\epsilon _{n,i}$'s. We
then set $Y_{n,i}=\epsilon _{n,i}\left( \zeta _{i\Delta _{n}}\right) $,
which, conditional on the $\zeta $ process, is Poisson distributed with mean 
$\zeta _{i\Delta _{n}}$. This representation is a special case of (\ref{yss}%
), with $g(\zeta,\epsilon) = \epsilon(\zeta)$ being a time-change and $%
R_{n,i} = 0$. We also note that although the $\epsilon _{n,i}$'s are assumed
to be i.i.d., the $(Y_{n,i})_{i\in\mathcal{I}_n}$ series can be highly
persistent through its dependence on the stochastic mean process $\zeta _{t}$%
.\hfill $\square $

\bigskip

To further broaden the empirical scope, we consider another example
concerning the bid-ask spread of asset quotes. This example is
econometrically interesting because of its resemblance to the
discrete-choice models (e.g., probit and logit) commonly used for modeling
binary and multinomial data.

\bigskip

\noindent \textsc{Example 5 (Bid-Ask Spread)}. Let $Y_{n,i}$ be the bid-ask
spread of an asset at time $i\Delta _{n}$. For a liquid asset, the spread is
often maintained at 1 tick (e.g., 1 cent), but it may widen to several ticks
due to a higher level of asymmetric information or dealer's inventory cost.
For ease of exposition, we suppose that $Y_{n,i}$ is a binary variable
taking values in $\{1,2\}$, while noting that a multinomial extension is
straightforward. Motivated by the classical discrete-choice models, we model
the spread as $Y_{n,i}=1+1\left\{ \zeta _{i\Delta _{n}}\geq \epsilon
_{n,i}\right\} $, and suppose that the variables $(\epsilon _{n,i})_{i\in 
\mathcal{I}_{n}}$ are i.i.d.\ and independent of the $\zeta _{t}$ process.
With the CDF of $\epsilon _{n,i}$ denoted by $F_{\epsilon }(\cdot )$, we
have $\mathbb{P}\left( Y_{n,i}=2|\zeta _{i\Delta _{n}}\right) =F_{\epsilon
}(\zeta _{i\Delta _{n}})$. Evidently, upon redefining $\zeta _{t}$ as $%
F_{\epsilon }\left( \zeta _{t}\right) $, we can assume that $\epsilon _{n,i}$
is uniformly distributed on the $[0,1]$ interval without loss of generality.
This normalization in turn allows us to interpret $\zeta _{t}$ as the
stochastic propensity of a \textquotedblleft wide\textquotedblright\ spread,
which may serve as a measure of market illiquidity.\hfill $\square $

\bigskip

We now proceed to establish the asymptotic validity of the permutation test
for the hypotheses described in (\ref{hypo}) for discretely valued
observations; see Theorem \ref{thmd} below. Since the state-space
representation (\ref{yss}) holds with the residual term $R_{n,i}=0$ in the
examples above, it seems reasonable to avoid unnecessary redundancy by
restricting our analysis to a simpler version given by%
\begin{equation}
Y_{n,i}=g\left( \zeta _{i\Delta _{n}},\epsilon _{n,i}\right) ,\quad i\in 
\mathcal{I}_{n}.  \label{yss2}
\end{equation}
We replace Assumption \ref{as:smooth} with the following assumption, where
we recall that for each $z\in \mathcal{Z}$, $F_{z}\left( \cdot \right) $
denotes the CDF of the random variable $g\left( z,\varepsilon _{n,i}\right) $
and $\tilde{\zeta}_{t}=\zeta _{t}-\Delta \zeta _{\tau }1\{t\geq \tau \}$.

\begin{assumption}
\label{as:sd}There exist a sequence $\left( T_{m}\right) _{m\geq 1}$ of
stopping times increasing to infinity, a sequence of compact subsets $(%
\mathcal{K}_{m})_{m\geq 1}$ of $\mathcal{Z}$, and a sequence $\left(
K_{m}\right) _{m\geq 1}$ of constants such that for each $m\geq 1$: (i) $%
\mathbb{P}\left( g\left( z,\epsilon _{n,i}\right) \neq g(z^{\prime
},\epsilon _{n,i})\right) \leq K_{m}\left\Vert z-z^{\prime }\right\Vert $
for all $z,z^{\prime }\in \mathcal{K}_{m}$; (ii) $\zeta _{t}$ takes values
in $\mathcal{K}_{m}$ for all $t\leq T_{m}$, and $\Vert \tilde{\zeta}%
_{t\wedge T_{m}}-\tilde{\zeta}_{s\wedge T_{m}}\Vert _{2}\leq K_{m}\left\vert
t-s\right\vert ^{1/2}$ for all $t,s$ in some fixed neighborhood of $\tau $.
\end{assumption}

\begin{theorem}
\label{thmd} In the state-space model (\ref{yss2}), suppose that Assumptions %
\ref{as:ind}(i), \ref{as:ind}(iii) and \ref{as:sd} hold, and that $%
k_{n}^{3}\Delta _{n}=o(1)$. Then, the following statements hold for the
permutation test $\hat{\phi}_{n}$ described in Algorithm 1:

(a) Under the null hypothesis in (\ref{hypo}), i.e., $\Delta \zeta _{\tau }=0
$, we have $\mathbb{E}[\hat{\phi}_{n}]\rightarrow \alpha $;

(b) Under a fixed alternative hypothesis in (\ref{hypo}), i.e., $\Delta
\zeta _{\tau }=c$ for some (unknown) constant $c\neq 0$, {and if $k_{n}\rightarrow \infty $, we have that $\mathbb{E}[\hat{\phi}_{n}]\rightarrow 1$.}
\end{theorem}

Theorem \ref{thmd} depicts the same asymptotic behavior of the permutation
test as in Theorem \ref{thmc}. The sufficient conditions of these results
differ mainly in how to gauge the closeness between the data and the
coupling variable, as manifest in the difference between Assumption \ref%
{as:smooth}(i) and Assumption \ref{as:sd}(i). The latter is easy to verify
under more primitive conditions in concrete settings. Specifically, in
Example 4, we note that $|g(z,\epsilon _{n,i})-g(z^{\prime },\epsilon
_{n,i})|$ is a Poisson random variable with mean $\left\vert z-z^{\prime
}\right\vert $, and hence, $\mathbb{P}(g(z,\epsilon _{n,i})\neq g(z^{\prime
},\epsilon _{n,i}))=1-\exp (-\left\vert z-z^{\prime }\right\vert )\leq
\left\vert z-z^{\prime }\right\vert $ as desired. In Example 5, we can use $%
\epsilon _{n,i}\sim \text{Uniform}\left[ 0,1\right] $ to deduce that 
\begin{equation*}
\mathbb{P}\left( g\left( z,\epsilon _{n,i}\right) \neq g\left( z^{\prime
},\epsilon _{n,i}\right) \right) =\mathbb{P}\left( 1\{z\geq \epsilon
_{n,i}\}\neq 1\{z^{\prime }\geq \epsilon _{n,i}\}\right) =\left\vert
z-z^{\prime }\right\vert ,
\end{equation*}%
which, again, verifies Assumption \ref{as:sd}(i). Therefore, in the context
of Examples 4 and 5 above, the permutation test is asymptotically valid for
detecting discontinuities in trading activity and illiquidity, respectively.


\section{Monte Carlo simulations\label{sec:mc}}


\subsection{Setting\label{sec:mc1}}

Our Monte Carlo experiment is based on the setting of Example 2. We simulate the (log) price process according
to $dP_{t}=\sigma _{t}dL_{t}$ under an Euler scheme on a 1-second mesh, and
then resample the data at the $\Delta _{n}=1$ minute frequency. We simulate $%
L$ either as a standard Brownian motion or as a (centered symmetric) stable
process with index $\beta =1.5$. To avoid unrealistic price path, we
truncate the stable distribution so that its normalized increment $\Delta
_{n}^{-1/\beta }\left( L_{i\Delta _{n}}-L_{(i-1)\Delta _{n}}\right) $ is
supported on $\left[ -C,C\right] $, and we consider $C\in \{10,20,30\}$ to
examine the effect of the support. The unit of time is one day. 

To simulate the volatility process, we first simulate two volatility factors
according to the following dynamics (see \cite{BT}):%
\begin{eqnarray*}
dV_{1,t} &=&0.0116(0.5-V_{1,t})dt+0.1023\sqrt{V_{1,t}}\left( \rho dL_{t}+%
\sqrt{1-\rho ^{2}}dB_{1,t}\right) +c\cdot 1_{\left\{ t=\tau \right\} }, \\
dV_{2,t} &=&0.6930(0.5-V_{2,t})dt+0.7909\sqrt{V_{2,t}}\left( \rho dL_{t}+%
\sqrt{1-\rho ^{2}}dB_{2,t}\right) +c\cdot 1_{\left\{ t=\tau \right\} },
\end{eqnarray*}%
where $B_{1,t}$ and $B_{2,t}$ are independent standard Brownian motions that
are also independent of $L_{t}$, $\rho =-0.7$ captures the negative correlation
between price and volatility shocks (namely the \textquotedblleft
leverage\textquotedblright\ effect). { The $V_{1}$ volatility factor is highly
	persistent with a half-life of 2.5 months, while the $V_{2}$ volatility
	factor is quickly mean-reverting with a half-life of only one day.} The constant $c$ determines the
size of the volatility jump at the event time $\tau $. In particular, $c=0$
corresponds to the null hypothesis, and we consider a range of $c$ values in $%
(0,5]$ in order to trace out a power curve for the corresponding alternative
hypotheses. The range of the $c$ parameter is calibrated according to %
\citeauthor{fomc}'s\ (\citeyear{fomc}) empirical estimates for FOMC
announcements.\footnote{%
Specifically, \cite{fomc} estimate the average jump size of $\log (\sigma
_{t})$ for the S\&P 500 ETF around FOMC announcements to be 1.037 (see Table 3
of that paper). This suggests that $%
\sigma _{\tau }^{2}/\sigma _{\tau -}^{2}=$ $(\exp (1.037))^{2}\approx 8$ on average, corresponding
to $c\approx  3.5$ in this Monte Carlo design.}
%

We note that the two volatility factors, $V_{1}$ and $V_{2}$, capture the
slow- and fast-mean-reverting volatility dynamics, respectively, with the
former having \textquotedblleft smoother\textquotedblright\ sample paths than
the latter. With this in mind, we simulate $\sigma _{t}$ using two
models:%
\begin{equation}
\left\{ 
\begin{array}{ll}
\text{Model A:} & \sigma _{t}^{2}=2V_{1,t}, \\ 
\text{Model B:} & \sigma _{t}^{2}=V_{1,t}+V_{2,t}.%
\end{array}%
\right.  \label{dgp}
\end{equation}%
In finite samples, Model A features relatively smooth volatility paths,
which is close to the \textquotedblleft ideal\textquotedblright\ scenario
underlying the infill asymptotic theory. Meanwhile, Model B generates more
realistic, and rougher, sample path for $\sigma $, providing a nontrivial
challenge for the proposed inference theory.

We implement the permutation test at the 5\% significance level, with the window
size $k_{n}\in \{15,30,60,90\}$. The six-fold increase from the smallest window size to the largest one represents a considerable range that allows us to explore the robustness of the proposed test with respect to the $k_n$ tuning parameter.\footnote{We also implemented simulations in which the two subsamples, $\mathcal{I}_{1,n}$ and $\mathcal{I}_{2,n}$, have different sample sizes $k_{1,n},k_{2,n}\in \{15,30,60,90\}$. As anticipated in footnote \ref{foot:kn}, these results are quantitatively similar to those with a common sample size, i.e., $k_{1,n}=k_{2,n}$. These additional results are omitted for brevity but are available upon request.}  The critical value is computed as in Remark \ref{rem:RandomPerm} based on 1,000 i.i.d.\ permutations. For comparison, we also implement the standard (two-sided)
t-test based on (\ref{ex1:tstat}). Rejection frequencies are computed based on
{2,000} Monte Carlo trials.

\subsection{Results\label{sec:mc2}}

\begin{table}[t]
\caption{Rejection Rates under the Null Hypothesis}
\label{tab:size}\centering
\begin{tabularx}{.9\textwidth}{rrYYYYcYYYY}\toprule
		&		& \multicolumn{4}{c}{Permutation Test} &    &   \multicolumn{4}{c}{T-test} \\
		&      &  (1) &   (2)   &   (3)  &  (4)   &      &  (1) &   (2)   &   (3)  &  (4) \\
		\midrule
		\multicolumn{11}{c}{\textit{Model A: One-factor Volatility}} \\	
		$k_n = 15$   &      &  0.050 & 0.058 & 0.053 & 0.062   & & 0.011 & 0.017 & 0.021 & 0.032 \\
		$k_n = 30$   &      &   0.054 & 0.049 & 0.056 & 0.057  & & 0.031 & 0.042 & 0.044 & 0.063 \\
		$k_n = 60$   &      &    0.047 & 0.048 & 0.053 & 0.059 & & 0.046 & 0.053 & 0.055 & 0.077 \\
		$k_n = 90$   &      &   0.058 & 0.050 & 0.053 & 0.056  & & 0.048 & 0.056 & 0.074 & 0.091 \\
		\\
		\multicolumn{11}{c}{\textit{Model B: Two-factor Volatility}} \\
		$k_n = 15$   &      &   0.049 & 0.055 & 0.054 & 0.062  & & 0.014 & 0.017 & 0.023 & 0.037 \\
		$k_n = 30$   &      &   0.055 & 0.050 & 0.055 & 0.056  & & 0.037 & 0.041 & 0.051 & 0.073 \\
		$k_n = 60$   &      &  0.049 & 0.049 & 0.054 & 0.064   & & 0.091 & 0.077 & 0.078 & 0.102 \\
		$k_n = 90$    &      & 0.064 & 0.052 & 0.058 & 0.059   & & 0.136 & 0.101 & 0.117 & 0.147 \\
		\bottomrule
		\multicolumn{11}{p{.88\textwidth}}{\textit{Note: }This table presents rejection frequencies of the permutation test and the t-test under the null hypothesis $\sigma^2_{\tau-} = \sigma^2_{\tau}$. The significance level is fixed at 5\%. Column (1) corresponds to the case with $L$ being a standard Brownian motion, and columns (2)--(4) correspond to cases in which $L$ is truncated stable with index $1.5$ and truncation parameter $C \in \{10,20,30\}$. The rejection frequencies are computed based on 2,000 Monte Carlo trials.}
	\end{tabularx}
\end{table}

We first examine the size properties of the permutation test $\hat{\phi}_n$ and the t-test based on (\ref{ex1:tstat}).
Table \ref{tab:size} reports the rejection frequencies of these tests under
the null hypothesis (i.e., $c=0$) for various data generating processes.
Column (1) corresponds to the case with $L$ being a standard Brownian
motion, and columns (2), (3), and (4) report results when $L$ is a truncated
stable process with the truncation parameter $C=10$, $20$, and $30$,
respectively.

The top panel of the table shows results from Model A, where the volatility
is solely driven by the ``slow'' factor. Quite remarkably, the rejection
frequencies of the permutation test are very close to the 5\% nominal level
for all specifications of $L$ and, importantly, for a wide range of the window size $k_n$. In contrast, the rejection rates of the t-test appear to be
far more sensitive to the choice of $k_n$. 
As we increase $k_n$ from 15 to 90, the rejection rate increases from 1.1\% to 4.8\% when $L$ is a Brownian motion. A similar pattern emerges when $L$ is a truncated stable process, except that the rejection rates now exceed the nominal level and reach 7.4\% and 9.1\% in columns (3) and (4), respectively. It is relevant to note that the t-test is not formally justified when $L$ is not a Brownian motion.

The more challenging case is Model B with the two-factor volatility dynamics.
Looking at the bottom panel of Table \ref{tab:size}, we find that the
permutation test still has rejection rates that are quite close to the
nominal level, although we see a slight over-rejection of 6.4\% when $k_n = 90$. This is
likely due to the fact that the approximation error in the coupling has 
nontrivial impact when the window size is large. That being said, we note that the benchmark t-test is more severely affected by this bias issue, with rejection rates reaching 9.1\% and 13.6\% when $k_n = 60$ and $k_n = 90$, respectively.

Figure \ref{fig:power} plots the power curves of the permutation test and the t-test for various $k_{n}$'s in Model A and Model B. We note that the four specifications of $L$ produce qualitatively similar results. We see that the rejection frequencies increase with the window size $k_{n}$ and the jump size $c$, which is expected from our consistency result obtained under $k_n\to \infty$. The permutation test appears to be less powerful than the t-test under the alternative hypothesis. This is a natural consequence of the typical trade-off between efficiency and robustness to the assumptions. The t-test is based on the spot variance estimator, which is ``locally'' the maximum-likelihood estimator of the spot variance under the Brownian shocks. The asymptotic validity and efficiency of the t-test rely on the Brownian assumption. In contrast, the permutation test is asymptotically valid regardless of whether the shocks are Brownian or not. As expected, this robustness is costly in terms of power. On the flip side, the power advantage of the t-test comes at the cost of size distortion when shocks are non-Brownian, which can be large, as shown in Table \ref{tab:size}.



\begin{figure}[t]
\includegraphics[width=\textwidth]{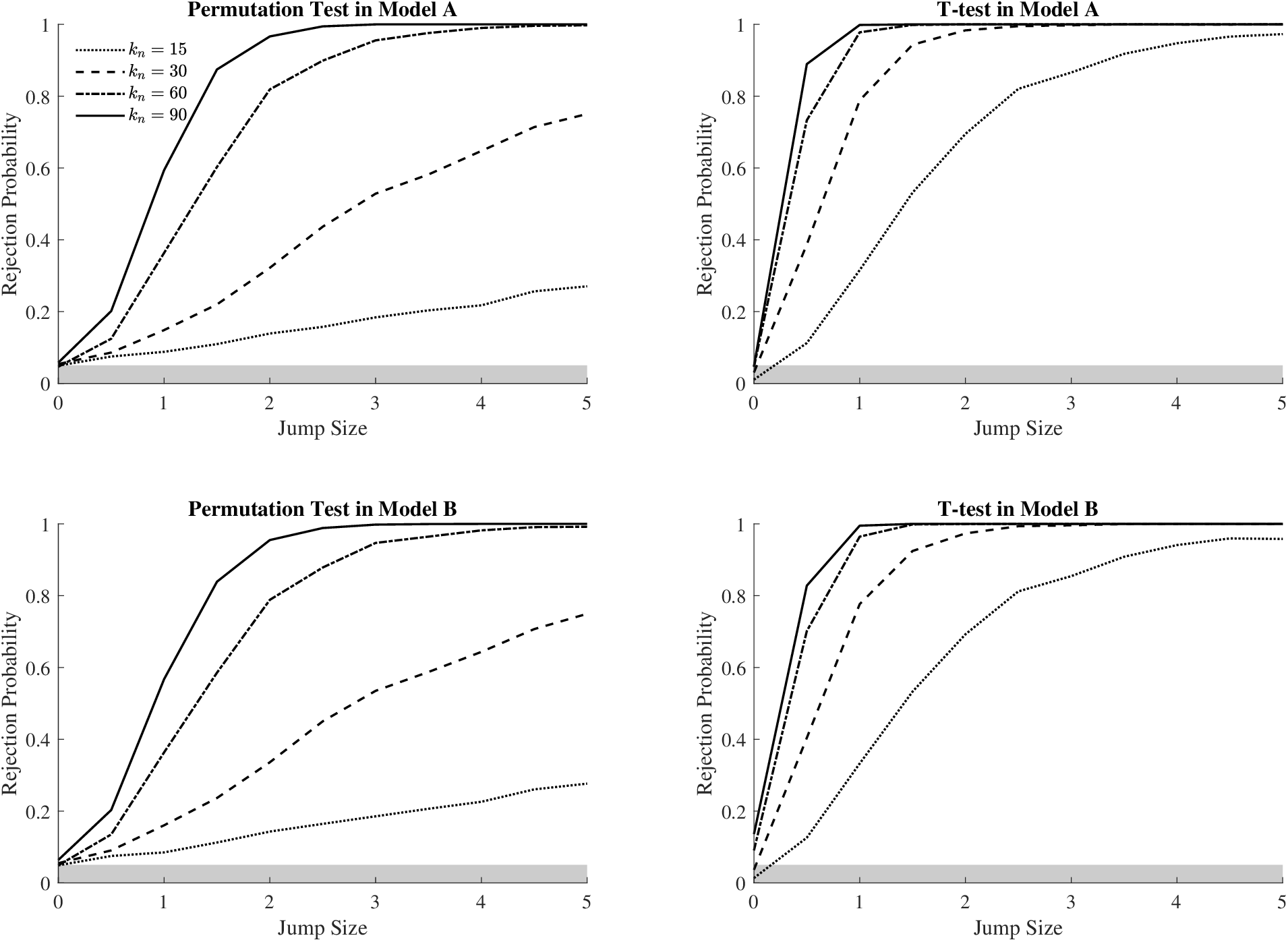}
\caption{The figure plots the rejection frequencies of the permutation test
and the t-test. The significance level is fixed at 5\% (highlighted by
shade). Results for Model A and Model B are presented in the top and bottom
rows, respectively. In all the plots, $L$ is simulated as a standard
Brownian motion.
The power curves are computed for the jump size parameter $%
c\in\{0,0.5,1,\ldots,5\}$. The rejection frequencies are computed based
on {2,000} Monte Carlo trials.}
\label{fig:power}
\end{figure}

Overall, we find that the permutation test controls size remarkably well
under the null hypothesis. Although it appears to be less powerful than the
t-test, it does not suffer from the latter's size distortion which can be
severe in the two-factor volatility model. Our results suggest that, given its robustness, the
permutation test is a useful complement to the conventional
test based on spot estimation and asymptotic Gaussian approximation.

\section{An empirical illustration\label{sec:emp}} 

We apply the proposed permutation test to a recent sample of high-frequency price and volume observations of the S\&P 500 ETF (NYSE: SPY) as an empirical illustration. The sampling frequency is one minute; the data source is the TAQ database. With the permutation test, we are interested in testing distributional discontinuities of the ETF's return, trading volume, and two measures of illiquidity for several FOMC announcements during the COVID-19 pandemic. This setting highlights one of the key merits of the proposed test, namely, it is applicable for a broad variety of high-frequency observations modeled in distinct ways. This is in sharp contrast to the conventional t-test based on \eqref{ex1:tstat}, which is designed specifically for testing volatility jumps and whose validity relies on the assumption of Brownian shocks.
We construct the high-frequency volume series as the total number of shares within each one-minute trading session. The illiquidity measures of interest include Amihud's measure defined as the ratio between absolute return and {dollar} trading volume (\cite{amihud2002illiquidity}) and the bid-ask spread averaged within each one-minute trading session. 

We focus on four important FOMC announcements in the 2020--2021 sample period that are related to distinct aspects of the Federal Reserve's monetary policy during the COVID-19 pandemic. The first is the announcement made on March 3, 2020, which was also the first FOMC announcement after COVID-19 hit the United States. The Fed stated its decision to lower the federal funds rate by 1/2 percentage point as its first response to counter the pandemic's negative impact on the economy. The second event occurred on December 16, 2020. At that time, being concerned with the rising long-term yield, many market participants expected that the Fed might implement the so-called ``operation twist'' to tame the steepening of the yield curve. But this turned out not to be on the Fed's agenda, and so, the Fed's ``inaction'' may be deemed as a shock relative to the market's anticipation. The third case pertains to the announcement on March 17, 2021. During the press conference, the Fed Chairman suggested that the central bank would be unlikely to raise {the} rate in the next 2--3 years, which may be regarded as a forward guidance on the target rate. The final example is the announcement on September 22, 2021, when the Fed officially declared its intention to taper the large-scale asset purchase program.

Figure \ref{fig:sum} plots the asset return and trading volume of SPY over one-hour windows centered at these announcement times. For ease of comparison, we plot the return and volume data for the four events on the same scale. These plots immediately reveal the highly distinct market conditions at those times. This  highlights the usefulness of adopting a high-frequency event-study research design, which allows us to investigate each event separately, rather than pooling information across different announcements under a likely fragile homogeneity assumption. We also observe several interesting patterns regarding how the market {responds} to the ``lumpy'' information embedded in the announcements. We generally see a rise in trading activity after the announcement. The price also tends to fluctuate more in the post-announcement window, although the March 3, 2020 event may be an exception as the market was already quite volatile even before the announcement. 


\begin{figure}[t]
	\centering
	\includegraphics[width=0.9\linewidth]{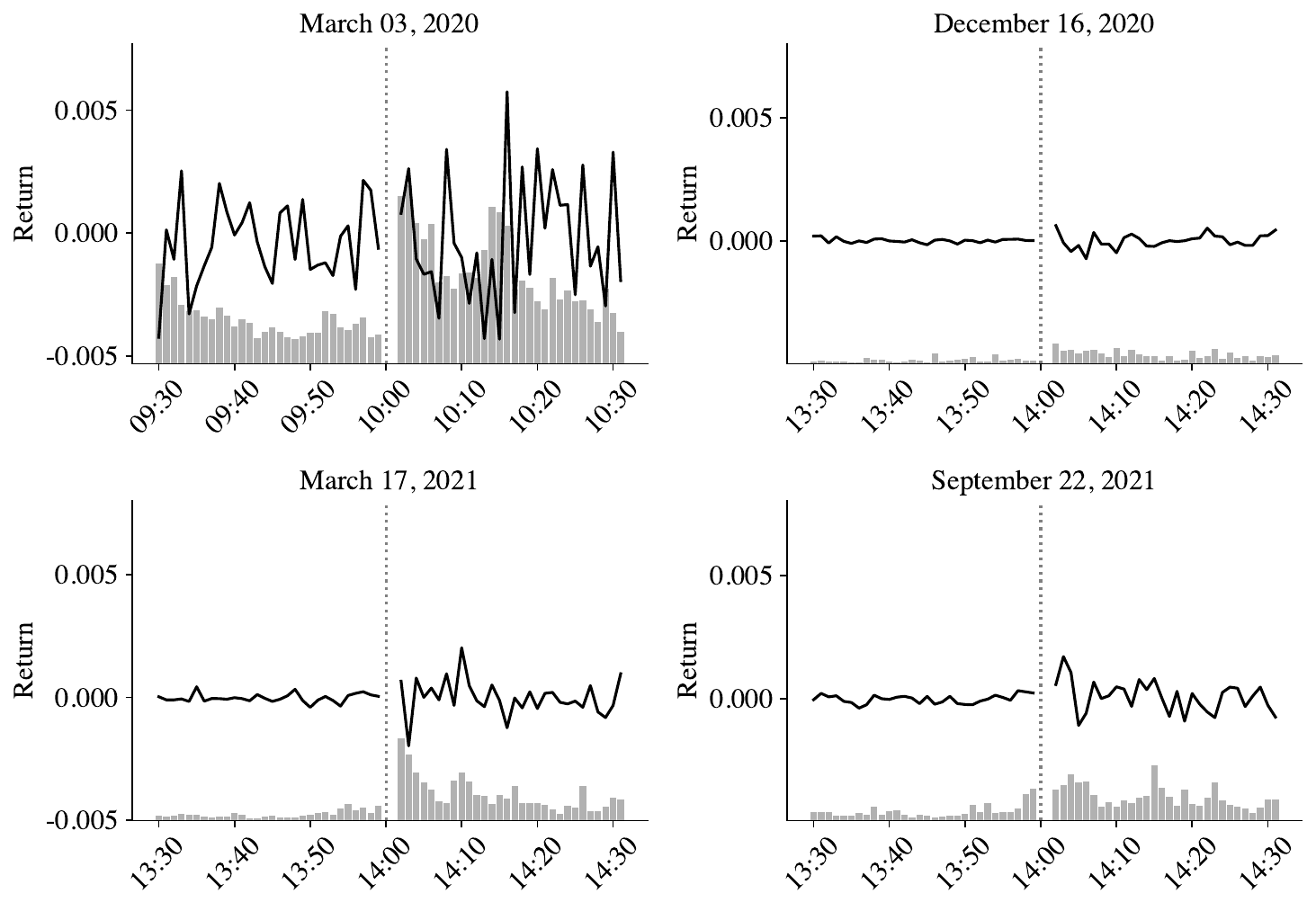}
	\caption{This figure plots the price {return} (solid) and volume (shaded) of the SPY ETF in one-hour windows around four FOMC announcements.}
	\label{fig:sum}
\end{figure}

As mentioned above, we implement the permutation test {described in Algorithm 1} on the ETF's returns, trading volume, Amihud's measure, and bid-ask spread constructed on the one-minute sampling frequency. For ease of interpretation, we consider the non-randomized version of the test described in Remark \ref{rem:RandomPerm}. We consider two event windows: $k_n = 10$ {minutes} or $30$ {minutes}. Recall that each FOMC announcement contains two stages. The first is an immediate release of a short summary on the Federal Reserve's webpage, which will be further detailed in the Fed {Chairman}'s opening statement during the first (roughly) 10 minutes of the press conference. The second part is a Q\&A session in which the {Chairman} {responds} to questions from the media, with the first few generally being more important. Given this setup, the shorter $k_n=10$ window allows us to focus on the immediate impact of the FOMC statement, whereas the longer $k_n=30$ window further covers the ``more subtle'' policy information conveyed to the public during the Q\&A session. It is worth noting that the 30-minute event window is also adopted in prior work on the high-frequency identification of monetary policy shocks; see \cite{nakamura2018QJE}.

\begin{table}[t]
	\centering
	\caption{Test Results for Selected FOMC Announcements}
	\label{tab:EmpTestResults}
	\begin{tabular}{ccccccccc}
		\toprule
		& \multicolumn{4}{c}{$k_{n}=10$} & \multicolumn{4}{c}{$k_{n}=30$}\\
		\cmidrule(r){2-5}\cmidrule(l){6-9}
  Date ($\tau$) & Return & Volume & Amihud & Spread & Return & Volume & Amihud & Spread\\
\midrule
03/03/2020 & & 
\textasteriskcentered{}\textasteriskcentered{}\textasteriskcentered{} &
\textasteriskcentered{} &
\textasteriskcentered{}\textasteriskcentered{}\textasteriskcentered{} &
& 
\textasteriskcentered{}\textasteriskcentered{}\textasteriskcentered{} &
&
\textasteriskcentered{}\textasteriskcentered{}\textasteriskcentered{}\\
12/16/2020 & 
\textasteriskcentered{}\textasteriskcentered{} &
\textasteriskcentered{}\textasteriskcentered{}\textasteriskcentered{} &
\textasteriskcentered{}\textasteriskcentered{} &
\textasteriskcentered{}\textasteriskcentered{}\textasteriskcentered{} &
\textasteriskcentered{}\textasteriskcentered{} &
\textasteriskcentered{}\textasteriskcentered{}\textasteriskcentered{} &
&
\textasteriskcentered{}\textasteriskcentered{}\textasteriskcentered{}\\
03/17/2021 & 
\textasteriskcentered{} &
\textasteriskcentered{}\textasteriskcentered{}\textasteriskcentered{} &
&
\textasteriskcentered{}\textasteriskcentered{}\textasteriskcentered{} & 
\textasteriskcentered{}\textasteriskcentered{} &
\textasteriskcentered{}\textasteriskcentered{}\textasteriskcentered{} &
&
\textasteriskcentered{}\textasteriskcentered{}\textasteriskcentered{}\\
09/22/2021 & 
\textasteriskcentered{}\textasteriskcentered{} &
\textasteriskcentered{}\textasteriskcentered{}\textasteriskcentered{} &
&
\textasteriskcentered{}\textasteriskcentered{}\textasteriskcentered{} &
\textasteriskcentered{}\textasteriskcentered{}\textasteriskcentered{} &
\textasteriskcentered{}\textasteriskcentered{}\textasteriskcentered{} &
&
\textasteriskcentered{}\textasteriskcentered{}\textasteriskcentered{}\\
\bottomrule
		\multicolumn{9}{l}{
			\begin{minipage}{16cm} 
				\vspace{.5em}
				\small Note: The permutation test is implemented following Algorithm 1. For ease of interpretation, we consider the non-randomized version as described in Remark \ref{rem:RandomPerm} using 100,000 i.i.d.\ permutations.
				Rejections at the 10\%, 5\%, and 1\% significance levels are indicated by \textasteriskcentered{}, \textasteriskcentered{}\textasteriskcentered{}, 
				and \textasteriskcentered{}\textasteriskcentered{}\textasteriskcentered{},  respectively.
			\end{minipage}
		}
	\end{tabular}
\end{table}

Table \ref{tab:EmpTestResults}  reports the results of the permutation test.  Recall that the permutation test applied to asset returns may be interpreted as a test for volatility jumps, as explained in Example 1 (with Brownian shocks) or Example 2 (with more general L\'evy shocks).\footnote{Prior work on volatility jumps (see, e.g., \cite{fomc} and \cite{litodorovzhang2021}) focuses exclusively on the case in which volatility is the scaling factor of Brownian shocks. Our permutation test is related to the prior work, but is valid under a more general notion of volatility with respect to non-Brownian shocks.} At a significance level of 10\%, the permutation test rejects the null of continuity for all events except the one on March 3, 2020. The latter non-rejection is consistent with our previous observation that the volatility of SPY was high even before the announcement. Meanwhile, the test also strongly rejects the null hypothesis of distributional continuity for the volume series, echoing the burst of trading activity seen in Figure \ref{fig:sum}.

The permutation test applied to the two illiquidity measures generates mixed results. We find some moderate evidence for the distributional discontinuity of Amihud's measure shortly after the March 3 and December 16 announcements in 2020. For the case with {a} 30-minute window, we do not reject the null of continuity for any of the announcements. The overall evidence {suggests} that the FOMC announcements under study did not lead to {an} abrupt change in the market impact coefficient gauged by Amihud's measure. Needless to say, this finding per se does not imply that the liquidity condition is unchanged after the announcement, as the notion of liquidity is a multifaceted concept. Indeed, we see that the permutation test applied to the bid-ask spread always strongly rejects the null of continuity. The post-announcement spread tends to be larger than its pre-announcement level, suggesting that it is significantly more costly to trade during {the} post-announcement trading session. 

All in all, the empirical illustration above demonstrates how the proposed permutation test may be used to test for distributional discontinuities for a variety of market variables. This type of versatility is not easily attained by existing methods in the high-frequency econometrics literature. We also see that interesting empirical findings may be obtained even with a small number of observations, which confirms the {practical} relevance of allowing the $k_n$ window to be possibly fixed in our asymptotic theory for the permutation test.
%


\section{Concluding remarks\label{conclusion}}

In this paper, we propose using a permutation test to detect discontinuities in an economic model at a cutoff point. Relative to the existing literature, we show that the permutation test is well suited for event studies based on time-series data. While nonparametric t-tests have been widely used for this purpose in various empirical contexts, the permutation test proposed in this paper provides a distinct alternative. Instead of relying on asymptotic (mixed) Gaussianity from central limit theorems, we exploit finite-sample properties of the permutation test in the approximating, or ``coupling'', two-sample problem.

We demonstrate that our new theory is broadly useful in a wide range of problems in the infill asymptotic time-series setting, which justifies using the permutation test to detect jumps in economic variables such as volatility, trading activity, and liquidity.
Compared with the conventional nonparametric t-test, the proposed permutation test has several distinct features. First, the permutation test provides asymptotic size control regardless of whether the sizes of the local subsamples are fixed or growing to infinity. In the latter case, we also establish that the permutation test is consistent. Second, the permutation test is versatile, as it can be applied without modification to many different contexts and under relatively weak conditions.

\begin{center}
	\textsc{Appendix: Proofs}
\end{center}


Throughout the proofs, we use $K$ to denote a positive constant that may
change from line to line, and write $K_{p}$ to emphasize its dependence on
some parameter $p$.  For any event $E\in \mathcal{F%
}$, we identify it with the associated indicator random variable.

\begin{proof}[Proof of Theorem \ref{thm1}]%
{Step 1.} Define $\phi _{n}$ in the same way as $\hat{\phi}_{n}$ but with $%
(Y_{n,i})_{i\in \mathcal{I}_{n}}$ replaced by $(U_{n,i})_{i\in \mathcal{I}%
_{n}}$. In this step, we show that 
\begin{equation}
\mathbb{E}[\hat{\phi}_{n}]=\mathbb{E}[\phi _{n}]+o(1).  \label{thm1:100}
\end{equation}

Let $\tilde{\phi}_{n}$ be defined in the same way as $\hat{\phi}_{n}$, but
with $( Y_{n,i}) _{i\in \mathcal{I}_{n}}$ replaced by $( \widetilde{Y}%
_{n,i})_{i\in \mathcal{I}_{n}}$, as defined in Assumption \ref{as:c}(iii).
Since $(\widetilde{Y}_{n,i})_{i\in \mathcal{I}_{n}}$ and $(Y_{n,i})_{i\in 
\mathcal{I}_{n}}$ {have} the same (conditional) distribution, 
\begin{equation}
\mathbb{E}[\tilde{\phi}_{n}]=\mathbb{E}[\hat{\phi}_{n}].  \label{thm1:101}
\end{equation}
Let $E_{n}\in\mathcal{F}$ be the event where the ordered values of $( U_{n,i}) _{i\in 
\mathcal{I}_{n}}$ and $(\widetilde{Y}_{n,i})_{i\in \mathcal{I}_{n}}$
correspond to the same permutation of $\mathcal{I}_{n}$. Since the test
statistic is only a function of the rank of the observations, we have $%
\tilde{\phi}_{n}=\phi _{n}$ in restriction to $E_n$. Hence, 
\begin{equation}
\vert \mathbb{E}[\tilde{\phi}_{n}]-\mathbb{E}[ \phi _{n}] \vert =\vert 
\mathbb{E}[\tilde{\phi}_{n} E_{n}^{c}]-\mathbb{E} [ \phi _{n} E_{n}^{c}]
\vert \leq \mathbb{P}( E_{n}^{c}) .  \label{thm1:102}
\end{equation}
By (\ref{thm1:101}) and (\ref{thm1:102}), (\ref{thm1:100}) follows from $%
\mathbb{P}( E_{n}^{c}) =o(1)$, which will be proved below.

Let $A_{n,i,j}\equiv \{U_{n,j}-U_{n,i}\geq 0,\widetilde{Y}_{n,j}-\widetilde{Y%
}_{n,i}<0\}$ for every $(i,j)\in \mathcal{I}_{n}\times \mathcal{I}_{n}$, and
note that $E_{n}^{c}\subseteq \cup _{i,j}A_{n,i,j}$. Recall the elementary
fact that if a sequence of random variables $X_{n}=o_{p}(1)$, then there
exists a real sequence $\delta _{n}=o(1)$ such that $\mathbb{P}(|X_{n}|\leq
\delta _{n})\to 1$. 
Under Assumption \ref{as:c}(iii), by applying this result to $%
X_{n}=2\max_{i\in \mathcal{I}_{n}}|\widetilde{Y}_{n,i}-U_{n,i}|k_{n}^{2}$,
we can find a sequence $\delta_{n}=o(1)$ such that 
\begin{equation}
\mathbb{P}\left( \max_{i\in \mathcal{I}_{n}}|\widetilde{Y}%
_{n,i}-U_{n,i}|\leq \delta_{n}k_{n}^{-2}/2\right) \rightarrow 1.
\label{thmrank:102}
\end{equation}%
We then observe that 
\begin{align*}
A_{n,i,j}& ~\subseteq ~\{U_{n,j}-U_{n,i}\geq \delta_{n}k_{n}^{-2},\widetilde{%
Y}_{n,j}-\widetilde{Y}_{n,i}<0\}\cup \{0\leq
U_{n,j}-U_{n,i}<\delta_{n}k_{n}^{-2}\} \\
& ~\subseteq ~\{|\widetilde{Y}_{n,j}-\widetilde{Y}_{n,i}-(U_{n,j}-U_{n,i})|>%
\delta_{n}k_{n}^{-2}\}\cup \{0\leq U_{n,j}-U_{n,i}<\delta_{n}k_{n}^{-2}\} \\
& ~\subseteq ~\{\max_{i\in \mathcal{I}_{n}}|\widetilde{Y}_{n,i}-U_{n,i}|>%
\delta_{n}k_{n}^{-2}/2\}\cup \{0\leq U_{n,j}-U_{n,i}<\delta_{n}k_{n}^{-2}\}.
\end{align*}%
Therefore, 
\begin{equation*}
E_{n}^{c}~\subseteq ~\cup _{i,j}A_{n,i,j}~\subseteq ~\{\max_{i\in \mathcal{I}%
_{n}}|\widetilde{Y}_{n,i}-U_{n,i}|>\delta_{n}k_{n}^{-2}/2\}\cup (\cup
_{i,j\in \mathcal{I}_{n}}\{0\leq U_{n,j}-U_{n,i}<\delta_{n}k_{n}^{-2}\}),
\end{equation*}%
which, together with (\ref{thmrank:102}), implies that 
\begin{equation}
\mathbb{P}(E_{n}^{c})\leq \mathbb{P}(\cup _{i,j\in \mathcal{I}_{n}}\{0\leq
U_{n,j}-U_{n,i}<\delta_{n}k_{n}^{-2}\})+o(1).  \label{thmrank:201}
\end{equation}%
Next, consider the following argument: 
\begin{align}
\mathbb{P}(\cup _{i,j\in \mathcal{I}_{n}}\{0\leq
U_{n,j}-U_{n,i}<\delta_{n}k_{n}^{-2}\}|\mathcal{G}_{n})& \leq \sum_{i,j\in 
\mathcal{I}_{n}}\mathbb{P}(0\leq U_{n,j}-U_{n,i}<\delta_{n}k_{n}^{-2}|%
\mathcal{G}_{n})  \notag \\
& \leq 2k_{n}\sum_{i\in \mathcal{I}_{n}}\sup_{x\in \mathbb{R}}\mathbb{P}%
(|U_{n,i}-x|\leq \delta_{n}k_{n}^{-2}|\mathcal{G}_{n})  \notag \\
& =O_{p}(\delta_n)=o_{p}(1),  \label{thmrank:302}
\end{align}%
where the last line holds by Assumption \ref{as:c}(ii). By (\ref{thmrank:302})
and the bounded convergence theorem, 
\begin{equation}
\mathbb{P}(\cup _{i,j\in \mathcal{I}_{n}}\{0\leq
U_{n,j}-U_{n,i}<\delta_{n}k_{n}^{-2}\})=o(1).  \label{thmrank:302b}
\end{equation}%
By combining (\ref{thmrank:201}) and (\ref{thmrank:302b}), we conclude that $%
\mathbb{P}(E_{n}^{c})=o(1)$, as desired.

{Step 2.} We now prove the assertions in parts (a) and (b) of the
theorem. In view of (\ref{thm1:100}), we only need to prove $\mathbb{E}[ \phi
_{n} ] \to \alpha $ and $\mathbb{E}[ \phi _{n}] \to 1$ in these two parts,
respectively. For part (a), note that $( U_{n,i}) _{i\in \mathcal{I}_{n}}$
are conditionally i.i.d.\ and so permutations constitute a group of
transformations that satisfy the randomization hypothesis in %
\citet[Definition 15.2.1]{lehmannromano2005}. Then, \citet[Theorem 15.2.1]{lehmannromano2005} implies that $\mathbb{E}[ \phi _{n}|\mathcal{\ G}_{n}]
=\alpha $, and $\mathbb{E}[ \phi _{n}] =\alpha $ then follows from the law
of iterated expectations.

To prove part (b), we need some additional notation. To emphasize the
dependence of $\widehat{T}_{n}$, $\widehat{T}_{n}^{\ast }$, and $\hat{\phi}
_{n}$ on the original data $(Y_{n,i})_{i\in \mathcal{I}_{n}}$, we explicitly
write them as $\widehat{T}_{n}( Y) $, $\widehat{T}_{n}^{\ast }( Y) $, and $%
\hat{\phi}_{n}( Y) $. With this notation, we can rewrite $\phi _{n}=\hat{\phi%
}_{n}( U) $, since it is computed in the same way as $\hat{\phi}_{n}$ but
with $(Y_{n,i})_{i\in \mathcal{I}_{n}}$ replaced by $(U_{n,i})_{i\in 
\mathcal{I}_{n}}$.

We first analyze the asymptotic behavior of $\widehat{T}_{n}( U) $. Define
the empirical analogue of $Q_{j,n}( \cdot ) $ as 
\begin{equation*}
\widehat{Q}_{j,n}( x) \equiv \frac{1}{k_{n}}\sum_{i\in \mathcal{I}
_{j,n}}1\{U_{n,i}\leq x\}.
\end{equation*}
Since the variables $(U_{n,i}) _{i\in \mathcal{I}_{j,n}}$ are $\mathcal{G}%
_{n}$-conditionally i.i.d., 
\begin{equation*}
\mathbb{E}[ ( \widehat{Q}_{j,n}( x) -Q_{j,n}(x)) ^{2}\vert \mathcal{G}_{n}]
\leq O(k_{n}^{-1}) = o(1).
\end{equation*}
By Markov's inequality and the law of iterated expectations, this implies
that $\widehat{Q}_{j,n}( x) -Q_{j,n}(x)=o_{p}(1)$ for each $x\in \mathbb{R}$%
. This and a classical Glivenko--Cantelli theorem (e.g., 
\citet[Theorem
21.5]{davidson:1994}) imply that 
\begin{equation}
\sup_{x\in \mathbb{R}}\vert \widehat{Q}_{j,n}( x) -Q_{j,n}(x)\vert =o_{p}(1),%
\text{ for }j\in \{ 1,2\} .  \label{thm1:801}
\end{equation}
By definition, 
\begin{equation*}
\widehat{T}_{n}( U) \equiv \frac{1}{2k_{n}}\sum_{i\in \mathcal{I} _{n}}( 
\widehat{Q}_{1,n}( U_{n,i}) -\widehat{Q}_{2,n}( U_{n,i}) ) ^{2}.
\end{equation*}
In addition, we define 
\begin{equation*}
S_{n}\equiv \frac{1}{2k_{n}}\sum_{i\in \mathcal{I}_{n}}( Q_{1,n}( U_{n,i})
-Q_{2,n}( U_{n,i}) ) ^{2}.
\end{equation*}
Note that the functions $\widehat{Q}_{j,n}( \cdot ) $ and $Q_{j,n}( \cdot ) $
are uniformly bounded. Hence, by the triangle inequality and (\ref{thm1:801}%
), 
\begin{align}
\vert \widehat{T}_{n}( U) -S_{n}\vert &\leq \frac{1}{ 2k_{n}}\sum_{i\in 
\mathcal{I}_{n}}\vert ( \widehat{Q}_{1,n}( U_{n,i}) -\widehat{Q}_{2,n}(
U_{n,i}) ) ^{2}-( Q_{1,n}( U_{n,i}) -Q_{2,n}( U_{n,i}) ) ^{2}\vert  \notag \\
&\leq \frac{K}{k_{n}}\sum_{i\in \mathcal{I}_{n}}\sum_{j\in \{1,2\}}\vert 
\widehat{Q}_{j,n}( U_{n,i}) -Q_{j,n}( U_{n,i}) \vert =o_{p}(1).
\label{thm1:802}
\end{align}
Conditional on $\mathcal{G}_{n}$, the bounded random functions $Q_{1,n}(
\cdot ) $ and $Q_{2,n}( \cdot ) $ can be treated as deterministic functions.
Next, note that 
\begin{equation}
S_{n}=\frac{1}{2}\sum_{j\in \{ 1,2\} }\int ( Q_{1,n}( x) -Q_{2,n}( x) )
^{2}dQ_{j,n}( x) +o_{p}(1)=\int ( Q_{1,n}( x) -Q_{2,n}( x) ) ^{2}d\overline{Q%
}_{n}( x) +o_{p}(1),  \label{thm1:803}
\end{equation}
where the first equality holds by a law of large numbers for the
conditionally i.i.d.\ variables $( U_{n,i}) _{i\in \mathcal{I}_{j,n}}$ for $j=1,2$, and
the second equality holds by the definition of $\overline{Q}_{n}$. By
combining (\ref{thm1:802}) and (\ref{thm1:803}), we deduce that 
\begin{equation}
\widehat{T}_{n}( U) =\int ( Q_{1,n}( x) -Q_{2,n}( x) ) ^{2}d\overline{Q}%
_{n}( x) +o_{p}(1).  \label{thm1:804_pre}
\end{equation}
In turn, by \eqref{thm1:804_pre} and the condition in part (b), we conclude that
\begin{equation}
\mathbb{P}( \widehat{T}_{n}( U) >\delta _{n}) \to 1 \text{ for any }\delta _{n}=o( 1) . \label{thm1:804}
\end{equation}

Next, we analyze the asymptotic behavior of $\widehat{T}_{n}^{\ast }( U) $.
It is useful to consider the following representation of this variable. We
denote $U_{\tilde{\pi} }=( U_{n,\tilde{\pi} ( i) }) _{i\in \mathcal{I}_{n}}$%
, where $\tilde{\pi} $ is a random permutation of $\mathcal{I}_{n}$,
independent from the data, and is drawn uniformly from the set of all
permutations of $\mathcal{I}_{n}$. By definition, $\widehat{T} _{n}^{\ast }(
U) $ is the $1-\alpha $ quantile of $\widehat{T} _{n}( U_{\tilde{\pi}}) $,
conditional on the sample, where the randomness comes from the random
realization of $\tilde{\pi}$. To analyze the permutation distribution, we
construct an additional coupling sequence of $( U_{n,i}) _{i\in \mathcal{I}%
_{n}}$ following the method of \citet[Section
5.3]{chungromano2013}.  We note that their coupling construction does not require the null hypothesis to hold, and it is thus suitable for our current purposes. The result of their coupling construction is another
random sequence $(U_{n,i}^{\prime })_{i\in \mathcal{I }_{n}}$ such that (i) $%
U_{n,i}=U_{n,i}^{\prime }$ for all $i$ in some random subset $\mathcal{I}%
_{n}^{\prime }\subseteq \mathcal{I}_{n}$; (ii) the cardinality of $\mathcal{I%
}_{n}\setminus \mathcal{I}_{n}^{\prime }$, denoted $D_n$, satisfies $\mathbb{%
E}[ D_{n}] =O(k_{n}^{1/2})$; and (iii) $(U_{n,i}^{\prime
})_{i\in \mathcal{I}_{n}}$ are $\mathcal{G}_{n}$-conditionally i.i.d.\ with
marginal distribution $\overline{Q}_{n}$.

{For any fixed arbitrary permutation $\pi $ and for} $j\in \{ 1,2\} $, define 
\begin{equation*}
\widehat{Q}_{j,n}( x;\pi ) \equiv \frac{1}{k_{n}}\sum_{i\in \mathcal{I}%
_{j,n}}1\{U_{n,\pi ( i) }\leq x\}\quad \text{and}\quad \widehat{Q}%
_{j,n}^{\prime }( x;\pi ) \equiv \frac{1}{k_{n}} \sum_{i\in \mathcal{I}%
_{j,n}}1\{U_{n,\pi ( i) }^{\prime }\leq x\}.
\end{equation*}
By repeatedly using the triangle inequality, 
\begin{align}
&\vert \widehat{T}_{n}( U_{\pi }) -\widehat{T}_{n}( U_{\pi }^{\prime }) \vert
\notag \\
&=\frac{1}{2k_{n}}\left\vert \sum_{i\in \mathcal{I}_{n}} \bigg( \big( 
\widehat{Q}_{1,n}( U_{n,\pi ( i) };\pi ) -\widehat{Q}_{2,n}( U_{n,\pi ( i)
};\pi ) \big) ^{2}- \big( \widehat{Q}_{1,n}^{\prime }( U_{n,\pi ( i)
}^{\prime };\pi ) -\widehat{Q}_{2,n}^{\prime }( U_{n,\pi ( i) }^{\prime
};\pi ) \big) ^{2}\bigg)\right\vert  \notag \\
&\leq \frac{K}{k_{n}}\sum_{j\in \{ 1,2\} }\sum_{i\in \mathcal{I} _{n}}%
\big\vert \widehat{Q}_{j,n}( U_{n,\pi ( i) };\pi ) -\widehat{Q}%
_{j,n}^{\prime }( U_{n,\pi ( i) }^{\prime };\pi ) \big\vert  \notag \\
&\leq \frac{K}{k_{n}^{2}}\sum_{i,k\in \mathcal{I}_{n}}\big\vert 1\{U_{n,\pi
( k) }\leq U_{n,\pi ( i) }\}-1\{U_{n,\pi ( k) }^{\prime }\leq U_{n,\pi ( i)
}^{\prime }\}\big\vert  \notag \\
&\leq KD_{n}/k_{n}=o_{p}(1),  \label{thm1:830}
\end{align}
where the last inequality uses the fact that $( U_{n,i},U_{n,k})
=(U_{n,i}^{\prime },U_{n,k}^{\prime })$ if $( i,k) \in \mathcal{I}
_{n}^{\prime }\times \mathcal{I}_{n}^{\prime }$, and so the summation on the
previous line only has $( 2k_{n}) ^{2}-( 2k_{n}-D_{n}) ^{2}\leq 4k_{n}D_{n}$
bounded terms that can be different from zero; and the $o_p(1)$ statement
follows from $\mathbb{E}[ D_{n}] =O(k_{n}^{1/2})$, $k_{n}\to \infty $, and
Markov's inequality.

For any fixed arbitrary permutation $\pi $, $\widehat{T}_{n}( U_{\pi
}^{\prime }) $ is the Cram\'{e}r-von Mises statistic for the $\mathcal{G}%
_{n} $-conditionally i.i.d.\ variables $(U_{n,\pi ( i) }^{\prime })_{i\in 
\mathcal{I}_{n}}$. Hence, by a similar argument leading to (\ref{thm1:804}),
we have $\widehat{T}_{n}( U_{\pi }^{\prime }) =o_{p}(1)$. By combining this
with (\ref{thm1:830}), it follows that 
\begin{equation}
\widehat{T}_{n}( U_{\pi }) =o_{p}(1).  \label{thm1:830b}
\end{equation}
Since this result holds for any arbitrary fixed permutation $\pi $, it also
holds for any pair of permutations considered at random from the set of all
possible permutations of $\mathcal{I}_{n}$, independently from the data. By
elementary properties of stochastic convergence, this implies the so-called
Hoeffding's condition (e.g., \citet[Equation (15.10)]{lehmannromano2005}).
By this and \citet[Theorem 15.2.3]{lehmannromano2005}, the permutation
distribution associated with the test statistic $\widehat{T}_{n}( U) $,
conditional on the data, converges to zero in probability. As a corollary of
this, 
\begin{equation}
\widehat{T}_{n}^{\ast }( U) =o_{p}(1).  \label{thm1:832}
\end{equation}

From (\ref{thm1:804}) and (\ref{thm1:832}), it is easy to see that $\widehat{T}_{n}( U) > \widehat{T}_{n}^{\ast }( U)$ with probability approaching 1. This further implies that $\mathbb{E}[\phi_n] \to
1$, which, together with (\ref{thm1:100}) proves the assertion of part (b). 
\end{proof}%

\bigskip

\begin{proof}[Proof of Theorem \ref{thmc}]
(a) We prove the assertion of part (a) by applying Theorem \ref{thm1}(a). We
construct the coupling variable $U_{n,i}$ as follows: 
\begin{equation}
U_{n,i}=g(\zeta _{(i^{\ast }-k_{n})\Delta _{n}},\epsilon _{n,i}),\quad \text{%
for all }i\in \mathcal{I}_{n}\text{.}  \label{thmc:101}
\end{equation}%
We set $\mathcal{G}_{n}=\mathcal{F}_{(i^{\ast }-k_{n})\Delta _{n}}$. By
Assumption \ref{as:ind}, $(\epsilon _{n,i})_{i\in \mathcal{I}_{n}}$ are
i.i.d.\ and independent of $\mathcal{G}_{n}$. Since $\zeta _{(i^{\ast
}-k_{n})\Delta _{n}}$ is $\mathcal{G}_{n}$-measurable, the variables $\left(
U_{n,i}\right) _{i\in \mathcal{I}_{n}}$ are $\mathcal{G}_{n}$-conditionally
i.i.d. This verifies the condition in part (a) of Theorem \ref{thm1}, which
also implies Assumption \ref{as:c}(i). It remains to verify conditions (ii)
and (iii) in Assumption \ref{as:c}.

By a standard localization argument (see \citet[Section 4.4.1]{jacodprotter2012}), we can strengthen Assumption \ref{as:smooth} by
assuming $T_{1}=\infty $, $\mathcal{K}_{m}=\mathcal{K}$, and $K_{m}=K$ for
some fixed compact set $\mathcal{K}$ and constant $K>0$. In particular, $%
\zeta _{(i^{\ast }-k_{n})\Delta _{n}}$ takes values in the compact set $%
\mathcal{K}$. By Assumption \ref{as:ind}, it is then easy to see that the $%
\mathcal{G}_{n}$-conditional probability density of $U_{n,i}=g(\zeta
_{(i^{\ast }-k_{n})\Delta _{n}},\epsilon _{n,i})$ is uniformly bounded (and
it does not depend on $i$). This implies condition (ii) of Assumption \ref%
{as:c}.

Finally, we verify condition (iii) of Assumption \ref{as:c}. By Assumption %
\ref{as:ind}(i), for each $i\in \mathcal{I}_{n}$, $\varepsilon _{n,i}$ is
independent of $\mathcal{F}_{i\Delta _{n}}$. Since $\zeta _{i\Delta _{n}}$
and $\zeta _{\left( i^{\ast }-k_{n}\right) \Delta _{n}}$ are $\mathcal{F}%
_{i\Delta _{n}}$-measurable, we deduce from Assumption \ref{as:smooth}(i)
that%
\begin{equation}
\mathbb{E}[|g(\zeta _{i\Delta _{n}},\epsilon _{n,i})-g(\zeta _{(i^{\ast
}-k_{n})\Delta _{n}},\epsilon _{n,i})|^{2}|\mathcal{F}_{i\Delta _{n}}]\leq
Ka_{n}^{2}\Vert \zeta _{i\Delta _{n}}-\zeta _{(i^{\ast }-k_{n})\Delta
_{n}}\Vert ^{2}.  \label{thmc:110}
\end{equation}%
Note that under the null hypothesis with $\Delta \zeta _{\tau }=0$, the
processes $\zeta _{t}$ and $\tilde{\zeta}_{t}$ are identical. Hence, by
Assumption \ref{as:smooth}(ii) and (\ref{thmc:110}),%
\begin{equation*}
\left\Vert g(\zeta _{i\Delta _{n}},\epsilon _{n,i})-g(\zeta _{(i^{\ast
}-k_{n})\Delta _{n}},\epsilon _{n,i})\right\Vert _{2}\leq
Ka_{n}k_{n}^{1/2}\Delta _{n}^{1/2}.
\end{equation*}%
By the maximal inequality under the $L_{2}$ norm (see, e.g., 
\citet[Lemma
2.2.2]{vandervaartwellner}), we further deduce that
\begin{equation}
\left\Vert \max_{i\in \mathcal{I}_{n}}\left\vert g(\zeta _{i\Delta
_{n}},\epsilon _{n,i})-g(\zeta _{(i^{\ast }-k_{n})\Delta _{n}},\epsilon
_{n,i})\right\vert \right\Vert _{2}\leq Ka_{n}k_{n}\Delta _{n}^{1/2}.
\label{thmc:111}
\end{equation}%
Recall that $a_{n}k_{n}^{3}\Delta _{n}^{1/2}=o(1)$ by assumption. Hence, 
\begin{equation}
\max_{i\in \mathcal{I}_{n}}\left\vert g(\zeta _{i\Delta _{n}},\epsilon
_{n,i})-g(\zeta _{(i^{\ast }-k_{n})\Delta _{n}},\epsilon _{n,i})\right\vert
=o_{p}(k_{n}^{-2}).  \label{thmc:112}
\end{equation}%
Note that, by the definitions in (\ref{yss}) and (\ref{thmc:101}),%
\begin{equation}
Y_{n,i}-U_{n,i}=g(\zeta _{i\Delta _{n}},\epsilon _{n,i})-g(\zeta _{(i^{\ast
}-k_{n})\Delta _{n}},\epsilon _{n,i})+R_{n,i}.  \label{thmc:113}
\end{equation}%
Combining (\ref{thmc:112}), (\ref{thmc:113}), and Assumption \ref{as:smooth}%
(iii), we deduce that $\max_{i\in \mathcal{I}_{n}}\left\vert
Y_{n,i}-U_{n,i}\right\vert =o_{p}\left( k_{n}^{-2}\right) $, which verifies
Assumption \ref{as:c}(iii). We have now verified all the conditions needed
in Theorem \ref{thm1}(a), which proves the assertion of part (a) of Theorem %
\ref{thmc}.

(b) We prove the assertion of part (b) by applying Theorem \ref{thm1}(b).
Under the maintained alternative hypothesis, we have $\Delta \zeta _{\tau }=c
$ for some constant $c\neq 0$. The coupling variable now takes the following form%
\begin{equation}
U_{n,i}=\left\{ 
\begin{array}{ll}
g(\zeta _{(i^{\ast }-k_{n})\Delta _{n}},\epsilon _{n,i}) & i\in \mathcal{I}%
_{1,n}, \\ 
g(\zeta _{(i^{\ast }-k_{n})\Delta _{n}}+c,\epsilon _{n,i}) & i\in \mathcal{I}%
_{2,n}.%
\end{array}%
\right.  \label{thmc:201}
\end{equation}%
Under Assumption \ref{as:ind}, it is easy to see that, for each $j\in
\{1,2\} $, the variables $(U_{n,i})_{i\in \mathcal{I}_{j,n}}$ are $\mathcal{G%
}_{n}$-conditionally i.i.d., which verifies Assumption \ref{as:c}(i).

We now turn to the remaining conditions in Assumption \ref{as:c}. As in part
(a), we can invoke the standard localization procedure and assume that the $%
\zeta _{t}$ process takes value in a compact set $\mathcal{K}$. Note that 
\begin{equation*}
\zeta _{\tau }-(\zeta _{(i^{\ast }-k_{n})\Delta _{n}}+\Delta \zeta _{\tau
})=\zeta _{\tau -}-\zeta _{(i^{\ast }-k_{n})\Delta _{n}}=o_{p}(1),
\end{equation*}%
where the $o_{p}(1)$ statement follows from the fact that the $\zeta _{t}$
process is c\`{a}dl\`{a}g and $k_{n}\Delta _{n}\rightarrow 0$. Therefore, by
enlarging the compact set $\mathcal{K}$ slightly if necessary, we also have $%
\zeta _{(i^{\ast }-k_{n})\Delta _{n}}+c\in \mathcal{K}$ with probability
approaching 1. Then, we can verify Assumption \ref{as:c}(ii) following the
same argument as in part (a). The verification of Assumption \ref{as:c}(iii)
is also similar.

Finally, we verify the condition in Theorem \ref{thm1}(b) pertaining to the
conditional CDFs. Note that%
\begin{equation*}
Q_{1,n}(x)=F_{\zeta _{(i^{\ast }-k_{n})\Delta _{n}}}(x)\quad \text{and}\quad
Q_{2,n}(x)=F_{\zeta _{(i^{\ast }-k_{n})\Delta _{n}}+c}(x).
\end{equation*}%
It is then easy to see that%
\begin{align*}
2\int \left( Q_{1,n}(x)-Q_{2,n}(x)\right) ^{2}d\overline{Q}_{n}\left(
x\right) \geq \int \left( F_{\zeta _{(i^{\ast }-k_{n})\Delta
_{n}}}(x)-F_{\zeta _{(i^{\ast }-k_{n})\Delta _{n}}+c}(x)\right)
^{2}dF_{\zeta _{(i^{\ast }-k_{n})\Delta _{n}}}(x).
\end{align*}
Since $\zeta _{(i^{\ast }-k_{n})\Delta _{n}}$ takes values in the compact
set $\mathcal{K}$, Assumption \ref{as:ind}(iii) implies that the lower bound
in the above display is bounded away from zero. Hence, $\int
\left( Q_{1,n}(x)-Q_{2,n}(x)\right) ^{2}d\overline{Q}_{n}\left( x\right)
>\delta _{n}$ for any real sequence $\delta _{n}=o(1)$. We have now verified
all conditions for Theorem \ref{thm1}(b), which proves the assertion of part
(b) of Theorem \ref{thmc}. %
\end{proof}%

\bigskip

\begin{proof}[Proof of Theorem \ref{thm2}]%
This proof follows from similar arguments to those used to
prove Theorem \ref{thm1}. For the sake of brevity, we focus on the only
substantial difference, which is how we establish that $\mathbb{P}%
(E_{n}^{c})=o(1)$. Recall that $E_{n}$ denotes the event where the ordered values
of $(U_{n,i})_{i\in \mathcal{I}_{n}}$ and $(\widetilde{Y}_{n,i})_{i\in 
\mathcal{I}_{n}}$ correspond to the same permutation of $\mathcal{I}_{n}$.
In the case of this proof, this result follows from 
\begin{equation*}
\mathbb{P}(E_{n}^{c})~\leq~ \mathbb{P}(\cup _{i\in \mathcal{I}_{n}}\{%
\widetilde{Y}_{n,i}\neq U_{n,i}\})~\leq~ \sum_{i\in \mathcal{I}_{n}}\mathbb{P}(%
\widetilde{Y}_{n,i}\neq U_{n,i})~=~o(1),
\end{equation*}%
where the first inequality follows from $E_{n}^{c}\subseteq \cup _{i\in 
\mathcal{I}_{n}}\{\widetilde{Y}_{n,i}\neq U_{n,i}\}$ and the convergence
follows from the assumption that $\mathbb{P}(\widetilde{Y}_{n,i}\neq
U_{n,i})=o\left( k_{n}^{-1}\right) $ uniformly in $i\in \mathcal{I}_{n}$. %
\end{proof}%

\bigskip

\begin{proof}[Proof of Theorem \ref{thmd}]%

(a) We prove this assertion  by applying Theorem \ref{thm2}(a). We shall
verify the conditions in Theorem \ref{thm2} for $\widetilde{Y}_{n,i}=Y_{n,i}$%
, $U_{n,i}=g(\zeta _{\left( i^{\ast }-k_{n}\right) \Delta _{n}},\epsilon
_{n,i})$, and $\mathcal{G}_{n}=\mathcal{F}_{\left( i^{\ast }-k_{n}\right)
\Delta _{n}}$. By assumption, the variables $\left( \epsilon _{n,i}\right)
_{i\in \mathcal{I}_{n}}$ are i.i.d.\ and independent of $\mathcal{G}_{n}$.
Hence, the variables $(U_{n,i})_{i\in \mathcal{I}_{n}}$ are $\mathcal{G}_{n}$%
-conditionally i.i.d.

It remains to verify that $\mathbb{P}\left( Y_{n,i}\neq U_{n,i}\right)
=o\left( k_{n}^{-1}\right) $ uniformly in $i\in \mathcal{I}_{n}$. By repeating the localization argument used in the proof of Theorem \ref{thmc}, we can strengthen Assumption \ref{as:sd} with $T_{1}=\infty $
without loss of generality. In particular, $\zeta _{t}$ takes values in some
compact subset $\mathcal{K}\subseteq \mathcal{Z}$. Note that for each $%
i\in \mathcal{I}_{n}$, $\epsilon _{n,i}$ is independent of $\left( \zeta
_{i\Delta _{n}},\zeta _{(i^{\ast }-k_{n})\Delta _{n}}\right) $. By
Assumption \ref{as:sd}(i), we thus have $\mathbb{P}\left( Y_{n,i}\neq
U_{n,i}|\mathcal{G}_{n}\right) \leq K\left\Vert \zeta _{i\Delta _{n}}-\zeta
_{\left( i^{\ast }-k_{n}\right) \Delta _{n}}\right\Vert $. Then, by
Assumption \ref{as:sd}(ii), we further have $\mathbb{P}\left( Y_{n,i}\neq
U_{n,i}\right) \leq K\left( k_{n}\Delta _{n}\right) ^{1/2}$. The condition $%
\mathbb{P}\left( Y_{n,i}\neq U_{n,i}\right) =o\left( k_{n}^{-1}\right) $
then follows from $k_{n}^{3}\Delta _{n}=o(1)$. By Theorem \ref{thm2}(a), we
have $\mathbb{E}[\hat{\phi}_{n}]\rightarrow \alpha $ as asserted.

(b)  We prove this assertion  by applying Theorem \ref{thm2}(b). We
verify the conditions in Theorem \ref{thm2} for $\widetilde{Y}_{n,i}=Y_{n,i}$%
, $\mathcal{G}_{n}=\mathcal{F}_{\left( i^{\ast }-k_{n}\right) \Delta _{n}}$,
and 
\begin{equation*}
U_{n,i}=\left\{ 
\begin{array}{ll}
g(\zeta _{\left( i^{\ast }-k_{n}\right) \Delta _{n}},\epsilon _{n,i}) & 
\text{if }i\in \mathcal{I}_{1,n}, \\ 
g(\zeta _{\left( i^{\ast }-k_{n}\right) \Delta _{n}}+c,\epsilon _{n,i}) & 
\text{if }i\in \mathcal{I}_{2,n}.%
\end{array}%
\right.
\end{equation*}%
Following the same argument as in part (a), we see that $(U_{n,i})_{i\in 
\mathcal{I}_{j,n}}$ are $\mathcal{G}_{n}$-conditionally i.i.d.\ for each $%
j\in \{1,2\}$, and $\mathbb{P}(Y_{n,i}\neq U_{n,i})=o\left(
k_{n}^{-1}\right) $ uniformly in $i\in \mathcal{I}_{n}$. Assumption \ref%
{as:ind}(iii) also ensures that $\mathbb{P}(\int \left(
Q_{1,n}(x)-Q_{2,n}(x)\right) ^{2}d\overline{Q}_{n}\left( x\right) >\delta
_{n})\rightarrow 1$ for any real sequence $\delta _{n}=o(1)$. By Theorem \ref%
{thm2}(b), we have that $\mathbb{E}[\hat{\phi}_{n}]\rightarrow 1$, as asserted. %
\end{proof}%

\begin{center}
	\textsc{{Supplement: Extension to other test statistics}}
\end{center}
As explained in Remark \ref{rem:otherStat}, our main results extend beyond the Cram\'er-von Mises statistic in \eqref{eq:CVM}. We characterize the relevant class of test statistics by the following high-level assumption.

\begin{assumption}\label{ass:testStat}
$\widehat{T}_{n}\equiv \Psi _{n}( ( Y_{n,i}) _{i\in \mathcal{I}_{n}}) $, where $( \Psi _{n}) _{n\in \mathbb{N} }$ is a sequence of functions that satisfies the following conditions:

(a)  $\widehat{T}_{n}$ is a rank statistic, i.e., for any $( Y_{n,i}) _{i\in \mathcal{I}_{n}}$ and $( Y_{n,i}^{\prime }) _{i\in \mathcal{I}_{n}}$ with $sign( Y_{n,i}-Y_{n,j}) =sign( Y_{n,i}^{\prime }-Y_{n,j}^{\prime }) $ for all $i,j\in \mathcal{I}_{n}$, $\Psi _{n}( ( Y_{n,i}) _{i\in \mathcal{I}_{n}}) =\Psi _{n}( ( Y_{n,i}^{\prime }) _{i\in \mathcal{I}_{n}}) $.

(b)  For any $( Y_{n,i})_{i\in \mathcal{I}_{n}}$ and $( Y_{n,i}^{\prime }) _{i\in \mathcal{I}_{n}}$, $\vert \Psi _{n}( ( Y_{n,i}) _{i\in \mathcal{I}_{n}}) -\Psi _{n}( ( Y_{n,i}^{\prime })_{i\in \mathcal{I}_{n}}) \vert =O_{p}(D_{n}/k_{n}) $ with $D_{n}=\vert \{i\in \mathcal{I}_{n}:Y_{n,i}\not=Y_{n,i}^{\prime }\} \vert $.
\end{assumption}

Assumption \ref{ass:testStat} is satisfied for a large class of test statistics, which includes the Cram\'er-von Mises and Kolmogorov-Smirnov statistics. In fact, the proof of Theorem \ref{thm1} shows that the Cram\'er-von Mises statistic satisfies Assumption \ref{ass:testStat}, and an analogous argument can be used to extend this to the Kolmogorov-Smirnov statistic. We now briefly describe the assumption. Assumption \ref{ass:testStat}(a) is an essential ingredient to our methodology. In turn, Assumption \ref{ass:testStat}(b) is a mild regularity condition that limits the influence that a few sample observations can have on the test statistic, and is only required to establish our consistency result.

The result proves Theorem \ref{thm1} for any test statistic that satisfies Assumption \ref{ass:testStat}. Since Theorem \ref{thm1} is the key to all of the results in the paper, this effectively implies that our findings extend to the class of statistics characterized by Assumption \ref{ass:testStat}, as claimed in Remark \ref{rem:otherStat}.

\begin{theorem}\label{thm1_high}
Under Assumptions \ref{as:c} and \ref{ass:testStat} (instead of \eqref{eq:CVM}),

(a) If the variables $( U_{n,i}) _{i\in \mathcal{I}_{n}}$ have the same $\mathcal{G}_{n}$-conditional distribution, we have $\mathbb{E} [ \hat{\phi}_{n}] \to \alpha $.

(b) Let $\widehat{T}_{n}( U) $ denote the test statistic but applied to $( U_{n,i}) _{i\in \mathcal{I}_{n}}$ instead of $ ( Y_{n,i}) _{i\in \mathcal{I}_{n}}$. If $k_{n}\to \infty $ and $\mathbb{P}( \widehat{T}_{n}( U) >\delta _{n}) \to 1$ for any real sequence $\delta _{n}=o( 1) $, we have $\mathbb{E}[ \hat{\phi}_{n}] \to 1$.
\end{theorem}
\begin{proof}
This proof follows closely that of Theorem \ref{thm1}, which has two steps. Step 1 remains unchanged, as it only relies on Assumption \ref{as:c} and the fact that $\widehat{T}_{n}$ is a rank statistic, imposed in Assumption \ref{ass:testStat}(a). Part (a) of Step 2 also remains unchanged, as it is entirely based on Step 1. To complete this proof, it then suffices to cover the analog of part (b) of Step 2.


We begin by considering the asymptotic behavior of $\widehat{T}_{n}^{\ast }( U) $, i.e., the $1-\alpha $ quantile of $\widehat{T}_{n}( U_{\tilde{ \pi}}) $, conditional on the sample, where the randomness comes from the realization of $\tilde{\pi}$. As in the proof of Theorem \ref{thm1}, we rely on the coupling construction based on \citet[Section 5.3]{chungromano2013}, which produces a random sequence $ (U_{n,i}^{\prime })_{i\in \mathcal{I}_{n}}$ such that (i) $ U_{n,i}=U_{n,i}^{\prime }$ for all $i$ in some random subset $\mathcal{I} _{n}^{\prime }\subseteq \mathcal{I}_{n}$; (ii) $D_{n}=\vert \{ i\in \mathcal{I}_{n}:U_{n,i}\not=U_{n,i}^{\prime }\} \vert $ satisfies $\mathbb{E}[D_{n}]=O(k_{n}^{1/2})$; and (iii) $(U_{n,i}^{\prime })_{i\in \mathcal{I}_{n}}$ are $\mathcal{G}_{n}$-conditionally i.i.d. Then, for any fixed arbitrary permutation $\pi $,
\begin{equation}
|\widehat{T}_{n}(U_{\pi })-\widehat{T}_{n}(U_{\pi }^{\prime })|=\vert \Psi _{n}( ( U_{n,\pi ( i) }) _{i\in \mathcal{I} _{n}}) -\Psi _{n}( ( U_{n,\pi ( i) }^{\prime }) _{i\in \mathcal{I}_{n}}) \vert =O_{p}( D_{n}/k_{n}^{1/2+\varepsilon }) =o_{p}( 1) ,
\label{eq:high2}
\end{equation}
where the second equality relies on $\vert \{ i\in \mathcal{I} _{n}:U_{n,\pi ( i) }\not=U_{n,\pi ( i) }^{\prime }\} \vert \leq D_{n}$ and Assumption \ref{ass:testStat}(b), and the last equality relies on $\mathbb{E}[D_{n}]=O(k_{n}^{1/2})$, $ k_{n}\to \infty $, and Markov's inequality. 
We can then repeat the arguments in the proof of Theorem \ref{thm1} to conclude that $\widehat{T}_{n}(U_{\pi }^{\prime })=o_{p}( 1) $ for any fixed arbitrary permutation $\pi $. By this with \eqref{eq:high2}, we conclude that $\widehat{T} _{n}(U_{\pi })=o_{p}( 1) $. Since $\pi $ was arbitrarily chosen, it then follows that
\begin{equation}
\widehat{T}_{n}^{\ast }( U) =o_{p}( 1) . \label{eq:high3}
\end{equation}

From the assumption in part (b) and \eqref{eq:high3}, it is easy to see that $\widehat{T}_{n}( U) >\widehat{T}_{n}^{\ast }( U) $ with probability approaching one. We can then repeat the remaining arguments in the proof of Theorem \ref{thm1} to complete this proof.
\end{proof}

\bibliography{BIBper.bib}

\end{document}


\begin{titlepage}
	\title{Supplemental Material to\\Permutation Tests for High-frequency Event Study\footnote{This research is conducted while Li is a visiting professor at Yale University and the Cowles foundation. Qiushi Zhang provided excellent research assistance.}}
	
	\date{\today}
	
	\vspace{1.0cm}
	\author{\mbox{}\\\mbox{}\\Federico Bugni\thanks{Department of Economics, Duke University, Durham, NC 27708; e-mail: federico.bugni@duke.edu.}~~and~~Jia Li\thanks{Department of Economics, Duke University, Durham, NC 27708; e-mail: jl410@duke.edu.}
		\\
		\mbox{}\\\mbox{} }
\end{titlepage}

\maketitle

\begin{abstract}
\thispagestyle{empty}\noindent This supplement contains addition simulation results.\newline
\newline

\end{abstract}

\newpage

\setcounter{equation}{0} 
\setcounter{lemma}{0} 

\renewcommand{\thesection}{S\Alph{section}}
\renewcommand{\theequation}{S\Alph{section}.\arabic{equation}}
\renewcommand{\thetheorem}{S\Alph{section}\arabic{theorem}}
\renewcommand{\thelemma}{S\Alph{section}\arabic{lemma}}
\renewcommand{\theassumption}{S\Alph{section}\arabic{assumption}}
\renewcommand{\theproposition}{S\Alph{section}\arabic{proposition}}
\renewcommand{\thecorollary}{S\Alph{section}\arabic{corollary}}

\input{supp1}

{\small 
\bibliographystyle{econometrica}
\bibliography{strong,fbm}
}